\documentclass[sigconf, nonacm]{acmart}

\newcommand\vldbpagestyle{plain} 

\usepackage{makecell}
\usepackage{subcaption}
\usepackage{capt-of}
\usepackage{algorithm}
\usepackage[noEnd=true, indLines=true]{algpseudocodex}
\usepackage{amsthm}
\usepackage{bbm}
\usepackage{paralist}
\usepackage{caption}
\usepackage{enumitem}
\usepackage[normalem]{ulem}
\usepackage{graphicx}
\usepackage{amsmath}
\usepackage{pgfplots}
\pgfplotsset{compat=1.18}
\usepackage{booktabs}
\usepackage{multirow}
\usepackage{ulem}
\usepackage{xcolor}
\DeclareMathOperator*{\argmax}{arg\,max}

\setitemize{leftmargin=*}
\usepackage{tikz}
\usetikzlibrary{shapes.arrows}
\usepackage{framed}
\usepackage{amsthm,thmtools,xcolor}

\newcommand\frbox[2][draw=black!30]{%
    \tikz[baseline]\node[%
        inner ysep=0pt, 
        inner xsep=2pt, 
        anchor=text, 
        rectangle, 
        rounded corners=1mm,
        #1] {\strut#2};%
}

\newcommand\appref[2][haque2024stochastic]{
\ifextended
  Appendix~\ref{#2}%
\else
  \cite[Appendix~\ref{#2}]{#1}\relax%
\fi
}

\newcommand{\algname}{\textit{Stochastic SketchRefine}}

\newcommand{\LCVaR}{\text{L-CVaR}\xspace}

\newcommand{\revision}[1]{{\color{red} #1}}
\definecolor{mygreen}{RGB}{2,112,10}
\definecolor{mygray}{rgb}{0.5,0.5,0.5}
\definecolor{mymauve}{rgb}{0.58,0,0.82}
\definecolor{byzantine}{rgb}{0.74, 0.2, 0.64}

\newcommand{\revbc}[1]{{\color{black}#1}}
\newcommand{\reva}[1]{{\color{black}#1}}
\newcommand{\revb}[1]{{\color{black}#1}}
\newcommand{\revc}[1]{{\color{black}#1}}

\usepackage{tcolorbox}
\newtcolorbox{expbox}[2][]{
    float=htb,
    blend before title=colon hang,
    sharp corners,
    colback=white,
    colbacktitle=blue!30!black,
    title={#2},
    every float=\centering,
    halign=flush center,
    halign lower=flush left,
    lower separated=false,
    #1}

\usepackage[frozencache=true,cachedir=minted-cache]{minted} 

\definecolor{commentgreen}{RGB}{2,112,10}

\usepackage{xspace}
\usepackage{cprotect}
\usepackage{hyperref}
\usepackage{soul}
\newcommand{\ssearch}{\textsc{SummarySearch}\xspace}
\newcommand{\skr}{\textsc{SketchRefine}\xspace}
\newcommand{\cvarif}{\textsc{LinearCVarification}\xspace}
\newcommand{\lcsolve}{\textsc{RCL-Solve}\xspace}
\newcommand{\lcsolveS}{\textsc{RCL-Solve-S}\xspace}
\newcommand{\solvesketch}{\textsc{SolveSketch}\xspace}
\newcommand{\sskr}{\textsc{Stochastic SketchRefine}\xspace}
\newcommand{\dpart}{\textsc{DistPartition}\xspace}
\newcommand{\kd}{\textsc{KD-Trees}\xspace}
\newcommand{\pshading}{\textsc{Progressive Shading}\xspace}
\newcommand{\pvscan}{\textsc{PivotScan}\xspace}
\newcommand{\spaql}{\textsc{SPaQL}\xspace}

\newcommand{\Prob}{\mathbb{P}}
\newcommand{\omegaU}{{\overline{\omega}}}
\newcommand{\mult}{x}
\newcommand{\Pmax}{P_\text{max}}

\usepackage{listings}
\definecolor{dkgreen}{rgb}{0,0.6,0}
\definecolor{gray}{rgb}{0.5,0.5,0.5}
\definecolor{light-gray}{gray}{0.90}
\definecolor{mauve}{rgb}{0.58,0,0.82}
\usepackage{mathtools}

\newtheoremstyle{mystyle}%
{3pt}
{3pt}
{\itshape\color{teal}}
{}
{\bfseries\color{teal}}
{.}
{.5em}
{}

\theoremstyle{mystyle}

\theoremstyle{definition}

\newif\ifextended
 \extendedtrue

\begin{document}
\frenchspacing
\setcounter{section}{0}
\setcounter{page}{1}
\pagebreak
\title{\algname: Scaling In-Database Decision-Making under Uncertainty to Millions of Tuples}



\author{Riddho R.~Haque}
\affiliation{%
  \institution{University of Massachusetts Amherst}
  \country{}
}
\email{rhaque@cs.umass.edu}

\author{Anh L. Mai}
\affiliation{%
  \institution{NYU Abu Dhabi}
  \country{}
}
\email{anh.mai@nyu.edu}

\author{Matteo Brucato}
\affiliation{%
  \institution{Microsoft Research}
  \country{}
}
\email{mbrucato@microsoft.com}

\author{Azza Abouzied}
\affiliation{%
  \institution{NYU Abu Dhabi}
  \country{}
}
\email{azza@nyu.edu}

\author{Peter J. Haas}
\affiliation{%
  \institution{University of Massachusetts Amherst}
  \country{}
}
\email{phaas@cs.umass.edu}

\author{Alexandra Meliou}
\affiliation{%
  \institution{University of Massachusetts Amherst}
  \country{}
}
\email{ameli@cs.umass.edu}

\begin{abstract}
Decision making under uncertainty often requires choosing \emph{packages}, or bags of tuples, that collectively optimize expected outcomes while limiting risks. Processing \emph{Stochastic Package Queries} (SPQs) involves solving very large optimization problems on uncertain data. Monte Carlo methods create numerous \emph{scenarios}, or sample realizations of the stochastic attributes of \emph{all} the tuples, and generate packages with optimal objective values across these scenarios. The number of scenarios needed for accurate approximation---and hence the size of the optimization problem when using prior methods---increases with variance in the data, and the search space of the optimization problem increases exponentially with the number of tuples in the relation. Existing solvers take hours to process SPQs on large relations containing stochastic attributes with high variance. Besides enriching the SPaQL language to capture a broader class of risk specifications, we make two fundamental contributions toward scalable SPQ processing. First, we propose \emph{risk-constraint linearization} (RCL), which converts SPQs into Integer Linear Programs (ILPs) whose size is independent of the number of scenarios used. Solving these ILPs gives us feasible and near-optimal packages. Second, we propose \sskr, a divide and conquer framework that breaks down a large stochastic optimization problem into subproblems involving smaller subsets of tuples. Our experiments show that, together, RCL and \sskr produce high-quality packages in orders of magnitude lower runtime than the state of the art.

\end{abstract}

\maketitle

\pagestyle{\vldbpagestyle}



\vspace{-2mm}

\section{Introduction}\label{sec:intro}


\looseness-1
Many decision-making problems involve risk-constrained optimization over uncertain data~\cite{ahmed2008solving}.
Consider a stock portfolio optimization problem (Figure~\ref{fig:portfolio-builder}), where stock prices at future dates are uncertain and are simulated as stochastic processes (e.g., Geometric Brownian Motion~\cite{reddy2016simulating}). We choose which stocks to buy, how many shares of each, and when to sell them, with constraints on budget and risk of loss.
Such a problem can be expressed as a package query using \spaql---Stochastic Package Query Language---allowing for the benefits of in-database decision-making~\cite{BrucatoYAHM20}.
The example query in Figure~\ref{fig:portfolio-builder} requests a portfolio that costs less than \$1000 and for which the probability of losing more than \$10 is at most 5\%---the latter constraint is called a \emph{Value-at-Risk} (VaR) constraint.
The package result is a bag of tuples, i.e., a portfolio of stocks and a selling schedule, that satisfies the constraints while maximizing the expected gain.

\begin{figure}
    \centering
    \begin{minipage}{0.6\linewidth}
        {\scriptsize
        \begin{tabular}{lcccc}
\multicolumn{4}{l}{\bf Stock\_Investments}\\
\toprule
\multicolumn{1}{l}{\bf ID} & \multicolumn{1}{l}{\bf Stock} & \multicolumn{1}{l}{\bf Price} & \multicolumn{1}{l}{\bf Sell After} & \multicolumn{1}{l}{\bf Gain} \\
\midrule
1                      & AAPL              & 195.71                       & 0.5 days                       & ?                        \\

2                      & AAPL               & 195.71                       & 1 day                     & ?                        \\

3                      & AAPL               & 195.71                       & 1.5 days                     & ?                        \\
\ldots                    & \ldots                & \ldots                          & \ldots                       & \ldots                     \\
730                      & AAPL           & 195.71                       & 365 days                      & ?                        \\
731                      & MSFT              & 373.04                       & 0.5 days                      & ?                        \\
732                      & MSFT             & 373.04                       & 1 day                      & ?                        \\
\ldots                    & \ldots                & \ldots                          & \ldots                       & \ldots                     \\
1460                      & MSFT               & 373.04                       & 365 days                      & ? \\
\ldots                    & \ldots                & \ldots                          & \ldots                       & \ldots                     \\
\bottomrule
\end{tabular}
        }
    \end{minipage}
    \begin{minipage}{0.3\linewidth}
        {\scriptsize\textbf{Scenarios}}
        
        \includegraphics[width=\linewidth]{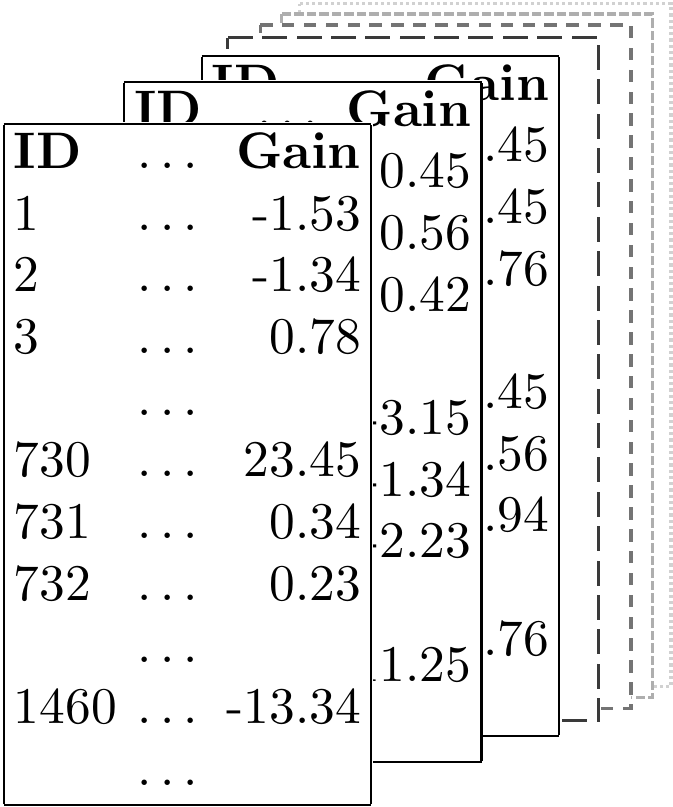}
    \end{minipage}

    \begin{framed}
    \vspace{-2mm}
    \begin{minipage}{0.48\linewidth}
    {\scriptsize\textbf{\spaql query}}
        {\scriptsize
        \input{queries/example_spaql_query}
        }
    \end{minipage}
    \begin{minipage}{0.18\linewidth}
        \begin{tikzpicture}
        \node[single arrow, draw=black, fill=white, 
          minimum width = 7pt, single arrow head extend=2pt,
          minimum height=10mm] {}; 
        \end{tikzpicture}
    \end{minipage}
    \begin{minipage}{0.3\linewidth}
        {\scriptsize
        \begin{tabular}{lcr}
\multicolumn{3}{l}{\bf Package result}\\
\toprule
\textbf{ID} & \ldots & \textbf{Count} \\
\midrule
3  & \ldots                       & 1                        \\
126  & \ldots                        & 2                        \\
1358 & \ldots                       & 1                        \\
2245 & \ldots                      & 1                        \\
\bottomrule
\end{tabular}
        }
    \end{minipage}
    \vspace{-4mm}
     \end{framed}
     \vspace{-3mm}
    \caption{The gain in the Stock\_Investments table is an uncertain attribute, simulated by stochastic processes.  The scenarios represent different simulations (possible worlds).  The example \spaql query contains a value-at-risk (VaR) constraint, specifying that the probability of total loss (negative gain) exceeding \$10 is at most 5\%.
    }
    \vspace{-5mm}
    \label{fig:portfolio-builder}
\end{figure}

To solve the stochastic optimization problem specified by a \spaql query, we approximate it by a deterministic problem via a Monte Carlo technique called Sample Average Approximation (SAA)~\cite{kim2015guide}. SAA creates numerous scenarios (``possible worlds''), each comprising a sample realization for every stochastic attribute of every tuple in the relation, e.g., a realized gain for every Stock--Sell\_After pair as in Figure~\ref{fig:portfolio-builder}. Such scenarios can be created using the Variable Generation (VG) functions of Monte Carlo databases like MCDB~\cite{JampaniXWPJH11}. 
\reva{
A VG function is a user-defined function that takes a table of input parameters and generates a table---a scenario---of sample values drawn from a corresponding probabilistic distribution. Because of their flexibility, VG functions can capture complex correlations across the sampled tuples.
In Figure~\ref{fig:portfolio-builder}, a VG function can use simulation models trained on historical stock prices to simulate discretized trajectories of future prices of every stock, and create scenarios of potential gains.} 

\reva{The scenarios are used to construct an SAA approximation to the original SPQ, where expectations are replaced by averages over scenarios and probabilities are replaced by empirical probabilities. E.g., the SAA approximation to the example query in Figure~\ref{fig:portfolio-builder} maximizes the average profit over all the scenarios, with at most 5\% of the scenarios having losses over \$10.}

Stochastic Package Queries (SPQs) are hard to scale on large datasets. They represent constrained optimization problems that grow along both:~(i)~\emph{the tuple dimension} and~(ii)~\emph{the scenario dimension}. More tuples directly correspond to more decision variables in an SPQ, causing the search space of the optimization problem---and correspondingly the solver runtime---to grow exponentially. On the other hand, increasing the number of scenarios is necessary for SAA accuracy, especially when stochastic attributes have high variance. However, this also increases the runtime 
for creating these scenarios, formulating and solving SAA optimization problems (whose size grows with the number of scenarios), and assessing the empirical risks of packages over the scenario set.
Even worse, having more tuples also necessitates having more scenarios due to the curse of dimensionality---more samples are needed to make the SAA sufficiently accurate
in high-dimensional decision spaces~\cite{campi2009scenario}.

\textbf{Prior work} tackled scaling along either the tuple or the scenario dimension, but not both. For all methods (including ours), we assume that the size of the optimal package is small, which is often the case for real-world problems.
\skr~\cite{brucato2018package} efficiently processes deterministic package queries (no scenarios involved) on large relations by partitioning the relations into groups of similar tuples and constructing representatives of each group offline. During query execution, the \emph{sketch} phase solves the problem over only the representatives; the \emph{refine} phase then iteratively replaces representatives in the sketch package with tuples from their partitions. Both sketch and refine queries are small enough ILPs that can be directly solved by off-the-shelf optimizers such as Gurobi.

\ssearch~\cite{BrucatoYAHM20} helps scale SPQ processing along the scenario dimension. It exploits the fact that SPQs 
can be approximated as ILPs using SAA. 
It creates packages using a set of \emph{optimization scenarios} and validates their feasibility over a much larger set of \emph{validation scenarios}. The key challenge is that large numbers of optimization scenarios are typically needed to mitigate the ``optimizer's curse'', which occurs when packages created from the optimization scenarios violate risk constraints among the validation scenarios; however, the size---and hence processing time---of the ILPs grows with the number of optimization scenarios. To handle this scaling issue, it replaces the optimization scenarios by a small set of conservative scenarios called \emph{summaries}. Constraints based on these summaries are harder to satisfy than constraints based on a random set of scenarios as in SAA, so packages created from them are more risk-averse and likely to be validation-feasible. \reva{E.g., the summary of a set of scenarios in our investment example might comprise the scenario-wise minimum gain for each tuple.} By varying the number of summaries, conservativeness can be carefully controlled to ensure near-optimal and feasible solutions. 

These methods have serious limitations. \skr does not work on stochastic data, and \ssearch does not scale well with more tuples because the formulated ILP has one decision variable per tuple in the relation. Moreover, \ssearch does not work well with high-variance stochastic attributes: As variance increases, more optimization scenarios are needed to accurately reflect the properties of the validation scenarios. With too few optimization scenarios, an intermediate package is likely to be validation-feasible but sub-optimal, so \ssearch will spend a lot of time increasing the number of summaries to relax the problem, only to discover that more scenarios are needed.

\textbf{Our work}
both enriches the expressiveness of the \spaql language and is the first to scale SPQ processing along both the tuple and the scenario dimensions. With respect to expressiveness, we extend SPaQL to allow specification of \emph{Conditional Value-at-Risk} (CVaR) constraints, also called ``expected shortfall'' constraints. \revc{The CVaR risk measure is widely used in finance, insurance, and other risk-management domains~~\cite{alexander2004comparison,McNeilFE15,Sarykalin2008}.} The VaR constraint in our portfolio example is of the form:
\begin{center}
\lstinline!SUM(Gain) <= -10 WITH PROBABILITY <= 0.05!
\end{center}
whereas a CVaR constraint might be of the form:
\begin{center}
\lstinline!EXPECTED SUM(Gain) >= -10 IN LOWER 0.05 TAIL!
\end{center}
Roughly speaking, the former constraint requires that the ``bad event'' where the loss exceeds \$10 occur with probability of at most 0.05, whereas the latter requires that, given the occurrence of a high-loss event of the form ``loss exceeds $\$x$'' having 0.05 probability, the expected value of the loss does not exceed \$10. (Section~\ref{sec:prob} gives precise definitions.)
Though VaR is a popular risk measure,  risk analysts often prefer CVaR over VaR because it is a \emph{coherent} risk measure with nice mathematical properties~\cite{Sarykalin2008}. E.g., 
portfolio optimization problems with CVaR constraints promote diversification of stocks, and CVaR can avoid extreme losses more effectively than VaR.
For example, \revc{given a portfolio that satisfies the VaR constraint, if an investor loses at least \$10 (which happens with at most  5\% probability), they can actually lose much more than \$10 without any way to control the expected loss. Having an additional CVaR constraint caps such expected losses when the bad event occurs.}

We achieve scaling through two novel mechanisms: \emph{risk-con\-st\-raint linearization} (RCL) and \sskr. 

\reva{\frbox{RCL}
handles scenario-scaling issues by replacing each VaR and CVaR constraint by a \emph{linearized CVaR} (\LCVaR) constraint (Section~\ref{sec:prob}).
%
%
Crucially, an SPQ with only \LCVaR constraints leads to a smaller SAA approximation whose size is independent of the number of scenarios and can be solved efficiently.
Like CVaR, an \LCVaR  constraint is parameterized by two values: a tail specifier $\alpha$ and a value $V$, e.g., $\alpha=0.05$ and $V=-10$ in the foregoing CVaR example.
Our new RCL procedure efficiently finds values of $\alpha$ and $V$ for each constraint such that the resulting feasible region for the SAA problem contains a solution that is both feasible and near optimal with respect to original SPQ.}


\frbox{\sskr} handles tuple-scaling issues by using a new data partitioning and evaluation mechanism.  It retains the divide-and-conquer structure of \skr---similarly leveraging off-the-shelf solvers and using sets of optimization and validation scenarios---but incorporates non-trivial extensions to overcome challenges raised by stochasticity. \sskr uses our novel, trivially-parallel\-iz\-a\-ble partitioning method, \dpart. Unlike existing hierarchical stochastic data clustering approaches~\cite{gullo2017information, huang2021robust}, \dpart's time complexity is sub-quadratic with respect to the number of tuples in the relation. This makes it suitable for partitioning million-tuple relations.

\noindent
\textbf{Contributions and outline.}  
RCL and \sskr synergistically reduce overall latencies, ensuring superior scalability over both tuples and scenarios. \sskr keeps high-variance tuples in separate partitions due to their dissimilarity with low-variance tuples. While refining partitions with high-variance tuples, use of RCL reins in the SPQ-processing latencies. 
In more detail, we organize our contributions as follows.

\begin{itemize}[leftmargin=*, topsep=2pt]
    \item We extend \spaql to allow for the expression of CVaR constraints, which allow for better risk control in stochastic optimization problems.  We further discuss necessary background and provide an overview of our approach and key insights. [Section~\ref{sec:prob}]
    \item We detail how RCL replaces VaR and CVaR constraints in an SPQ by \LCVaR constraints, while ensuring the resulting packages are feasible and near-optimal for the original SPQ. [Section~\ref{sec:cvarif}]
  
  \item We present \sskr, a two-phase divide-and-conquer approach that splits SPQs on large relations into smaller problems that can be solved quickly.  Similar to \skr, it first solves a \emph{sketch} problem using representatives over data partitions, and then \emph{refines} the initial solution with data from each partition.  In contrast with \skr, our method employs stochastically-identical \emph{duplicates} of representatives 
  to handle the challenges of stochasticity.   [Section~\ref{section:sskr}]
  \item \sskr relies on appropriately-partitioned data to generate the sketch and refine problems. We propose a novel offline partitioning method for stochastic relations, \dpart, which can effectively partition stochastic relations containing millions of tuples in minutes. [Section~\ref{sec:partitioning}]
  \item We show that \sskr achieves a $(1-\epsilon)^2$-optimal solution w.r.t a user-defined approximation error bound $\epsilon$. [\appref{sec:theory}] 
  \item We show via experiments that \sskr generates high quality packages in an order of magnitude lower runtime than \ssearch.
  Furthermore, \sskr can execute package queries over millions of tuples within minutes, whereas \ssearch fails to produce any package at such scales even after hours of execution. [Section~\ref{sec:experimental_evaluation}]
\end{itemize}

\section{CVaR and Related Constraints}
\label{sec:prob}

In this section, we discuss the extension of \spaql via the addition of CVAR constraints. 
We first briefly review the types of constraints previously supported by \spaql and then describe the syntax and semantics of CVaR and related constraints in detail.

\smallskip\noindent\textbf{\spaql constraint modeling.} \looseness-1
We assume a relation with $n$ tuples. In SPQ evaluation, each tuple $t_i$ is associated with an integer variable $x_i$ that represents the tuple's multiplicity in the package result. A \emph{feasible package} is an assignment of the variables in $x=(x_1, \dots, x_n)$, that satisfies all query constraints; as in the introductory example, we typically want to find a feasible package that maximizes or minimizes a given linear objective function; see \cite[Appendix~A]{BrucatoYAHM21} for a complete \spaql language description. We denote by $t_i.A$ the 
value of the attribute $A$ in tuple $t_i$; if $A$ is stochastic, then $t_i.A$ is a random variable. 
The \spaql constraint types are as follows.
\begin{itemize}[leftmargin=*]
    \item \colorbox{light-gray}{\small \lstinline!REPEAT R!}:
    \emph{Repeat constraints}  cap the multiplicities of every tuple in the package, i.e., $x_i \leq (1+R)$, $\forall 1 \leq i \leq n$.
    \item\colorbox{light-gray}{\small \lstinline!COUNT(Package.*) <= S!}:
    \emph{Package-size constraints} bound the total number of tuples in the package, i.e., $\sum_{i=1}^{n} x_i \leq S$.
    \item \colorbox{light-gray}{\small \lstinline!SUM(A) <= V!}:
    \emph{Deterministic sum constraints} bound the sum of the values of a deterministic attribute $A$ in the package, i.e., $\sum_{i=1}^{n}  t_i.A * x_i \leq V$.
    \item \colorbox{light-gray}{\small \lstinline!EXPECTED SUM(A) <= V!}:
    \emph{Expected sum constraints} bound the expected value of the sum of the values of a stochastic attribute $A$ in the package, i.e., $\mathbb{E}[\sum_{i=1}^{n} t_i.A*x_i] \leq V$. 
    \item \colorbox{light-gray}{\small \lstinline[mathescape]!SUM(A) <= V WITH PROBABILITY <= $\ \alpha$!}:
    \emph{VaR constraints} bound the probability that the sum of stochastic attribute $A$ is below or above a value, i.e., $\Prob(\sum_{i=1}^{n} t_i.A * x_i \leq V) \leq \alpha$.
\end{itemize}

\smallskip\noindent\textbf{Value-at-Risk.} 
We first define the notions of quantile and VaR. Consider a stochastic attribute $A$ with cumulative distribution function $F_A$. For $\alpha\in[0,1]$, we define the \emph{$\alpha$-quantile} of $A$---or, equivalently, of $F_A$---as
$
q_\alpha(A)=\inf\{y:F_A(y)\ge \alpha\}.
$
We can then define the \emph{$\alpha$-confidence Value at Risk} of $A$ as
$
\text{VaR}_\alpha(A)=q_\alpha(A),
$
i.e., the VaR is simply a quantile. Thus a VaR constraint as above can be interpreted as a constraint of the form $\text{VaR}_\alpha(\sum_{i=1}^nt_i.A*x_i)\ge V$. 

\smallskip\noindent\textbf{Conditional Value-at-Risk.}
As discussed in the introduction, CVaR constraints help control extreme risks beyond what VaR constraints can provide. Following \cite{McNeilFE15,Sarykalin2008}, we define the \emph{lower-tail $\alpha$-confidence Conditional Value-at-Risk} of $A$ as 
\begin{equation}\label{eq:defCVaRL}
\text{CVaR}^\bot_\alpha(A)=\frac{1}{\alpha}\int_0^\alpha q_u(A)\,du=
\frac{1}{\alpha}\int_0^\alpha \text{VaR}_u(A)\,du.
\end{equation}
and the \emph{upper-tail $\alpha$-confidence Conditional Value-at-Risk} of $A$ as
\[
\text{CVaR}^\top_\alpha(A)=\frac{1}{1-\alpha}\int_\alpha^1 q_u(A)\,du=
\frac{1}{1-\alpha}\int_\alpha^1 \text{VaR}_u(A)\,du.
\]
From \cite[Lemma~2.16]{McNeilFE15}, we have that if $F_A$ is continuous, then
\begin{align}
              &\text{CVaR}^\bot_\alpha(A)=\mathbb{E}[A\mid A\le q_\alpha(A)] \label{eq:defECVaRL}\\
    \text{and}\hspace{5ex}&\text{CVaR}^\top_\alpha(A)=\mathbb{E}[A\mid A\ge q_\alpha(A)], \label{eq:defECVaRU}
\end{align}
which motivates our CVaR-constraint syntax. Specifically, for a value such as $\alpha=0.05$, the following \spaql CVaR constraints
\begin{lstlisting}
EXPECTED SUM(A) >= @$V$@ IN LOWER @$\alpha$@ TAIL
EXPECTED SUM(A) <= @$V$@ IN UPPER @$\alpha$@ TAIL
\end{lstlisting} 
correspond to the constraints $\text{CVaR}^\bot_\alpha(\sum_{i=1}^nt_i.A * x_i)\ge V$ and $\text{CVaR}^\top_{1-\alpha}(\sum_{i=1}^nt_i.A * x_i)\le V$, respectively.
In a portfolio setting, constraints using $\text{CVaR}^\bot_\alpha(A)$ are useful when $A$ represents a gain as in our portfolio example, and constraints using $\text{CVaR}^\top_{1-\alpha}(A)$ are useful when $A$ represents a loss. Since $\text{CVaR}^\bot_\alpha(A)=\text{CVaR}^\top_{1-\alpha}(-A)$, 
without loss of generality,
we will focus on $\text{CVaR}^\bot_\alpha$ and simply denote it by $\text{CVaR}_\alpha$; we will also restrict attention to maximization problems where the use of $\text{CVaR}^\bot_\alpha$ makes sense. Also, for simplicity, we assume henceforth that all random attributes have continuous distributions so that the representations in \eqref{eq:defECVaRL} and \eqref{eq:defECVaRU} are valid.
Then, given a set $S$ of i.i.d.\ samples from the distribution $F_Z$ of a random variable $Z$, we can estimate $\text{CVaR}_\alpha(Z)$ simply as $\widehat{\text{CVaR}}_\alpha(Z)$, the average over the lowest $\alpha$-fraction of values in $S$.\footnote{In general, $\text{CVaR}^\bot_\alpha(Z)=\mathbb{E}[Z\mid Z\le q_\alpha(Z)]+q_\alpha(Z)\bigl[\alpha - \mathbb{P}\bigl(Z\le q_\alpha(Z)\bigr)\bigr]$, so the only modification to the methods in this paper is that estimation of CVaR from optimization scenarios becomes slightly more complex.}





\smallskip
\noindent
\looseness-1
\textbf{Scalably solving SPQs with linearized CVaR.} In contrast with VaR constraints, CVaR constraints are convex, naturally reducing the complexity of SPQs.   However, replacing each VaR constraint with CVaR---so the query contains only CVaR constraints---does not suffice to solve SPQs efficiently: The convex optimization problem resulting from the modified SPQ is not amenable to SAA approximation and is usually very expensive to solve. 
%
To scalably solve
SPQs, we introduce the notion of \emph{linearized CVaR constraints}. For a package represented by integer variables\footnote{In a slight abuse of terminology, we use the term ``package'' to refer interchangeably to either the integer vector $x$ or the bag of tuples specified by $x$.} $x=(x_1,\ldots,x_n)$---i.e., the multiplicities of tuples in the package---and a stochastic attribute $A$, define
\[
\LCVaR_\alpha(x,A)=\sum_{i=1}^n \text{CVaR}_\alpha(t_i.A)*x_i.
\]
An \LCVaR constraint has the form $\LCVaR_\alpha(x,A)\ge V$ (in the lower $\alpha$-tail). When identically parameterized, an $\LCVaR$ constraint is more restrictive than a CVaR constraint, which in turn is more restrictive than a VaR constraint. Formally, letting $x\cdot A$ denote $\sum_{i=1}^n t_i.A * x_i$ and $\mathbb{Z}_0$ denote the set of nonnegative integers, we have the following result, which holds even if $F_A$ is not continuous. 
See \appref{appendix:APSproof} for all proofs.
\begin{theorem}\label{th:constraintCompare}
For any $x\in \mathbb{Z}_0^n$, $\alpha\in[0,1]$, $V\in\mathbb{R}$, and stochastic attribute $A$:  
\[
\begin{split}
\LCVaR_\alpha(x,A)\ge V&\implies \text{CVaR}_\alpha(x\cdot A)\ge V\\
&\implies \text{VaR}_\alpha(x\cdot A)\ge V
\end{split}
\]
\end{theorem}

\reva{As a simple example, consider a relation with two tuples $t_1$ and $t_2$ and a stochastic attribute $A$ such that $\Prob(t_1.A=-1)=\Prob(t_1.A=1)=0.5$, and $t_2.A$ and $t_1.A$ are i.i.d. For $x=(1,1)$ and $\alpha=0.5$, we can verify that $\LCVaR_{0.5}(x,A)\le \text{CVaR}_{0.5}(x\cdot A)\le \text{VaR}_{0.5}(x\cdot A)$, which implies the assertion of the theorem for this example. Observe that $\text{CVaR}_{0.5}(t_1.A)=\text{CVaR}_{0.5}(t_2.A)=-1$, so that $\LCVaR_{0.5}(x,A)=1\cdot \text{CVaR}_{0.5}(t_1.A)+1\cdot \text{CVaR}_{0.5}(t_1.A)=-2$. Also, the random variable $Z=x\cdot A=1\cdot t_1.A+1\cdot t_2.A$ satisfies $\Prob(Z=-2)=\Prob(Z=2)=0.25$ and $\Prob(Z=0)=0.5$. Thus $\text{CVaR}_{0.5}(x\cdot A)=\mathbb{E}[Z\mid Z\le 0]=-2/3$ and  $\text{VaR}_{0.5}(x\cdot A)=\text{median}(Z)=0$, verifying the inequalities. Intuitively, the first inequality holds---strictly, in this example---because $\LCVaR$ involves separate averages over the lower tails of $t_1.A$ and $t_2.A$, each of which have half their probability mass strictly below the median at 0. In contrast, the random variable $Z$ has significant mass at the median value of 0, because positive values of $t_1.A$ can compensate for negative values of $t_2.A$ and vice versa, so the CVaR lower-tail average is higher. Moreover, VaR is simply the median of $Z$ whereas CVaR is an average of values at or below the median, which explains the second (strict) inequality.
}

\section{\lcsolve: Solving Small SPQs}\label{sec:cvarif}

In this section, we describe \lcsolve (Algorithm~\ref{alg:lcsolve}), a standalone scenario-scalable algorithm for solving SPQs over a relatively small set of tuples. By completely eliminating the need for scenario summaries, \lcsolve outperforms \ssearch. Moreover, our new \sskr algorithm (Section~\ref{section:sskr}) solves a large-scale SPQ with many tuples by solving a sequence $Q(S_0), Q(S_1),\ldots$ of relatively small-scale SPQs, 
and slight variants of Algorithm~\ref{alg:lcsolve} (discussed in Section~\ref{section:sskr}) can be used to solve each SPQ encountered during the sketch-and-refine process.

\begin{algorithm}[t]
\caption{\textsc{\lcsolve}}\label{alg:lcsolve}
{\footnotesize
\begin{algorithmic}[1]
\Require $Q :=$ A Stochastic Package Query
\Statex $S :=$ Set of stochastic tuples in the SPQ
\Statex $m :=$ Initial number of optimization scenarios
\Statex $\mathcal{V}$ $:=$ Set of validation scenarios
\Statex $\delta :=$ Bisection termination distance
\Statex $\epsilon :=$ Error bound

\Ensure $x :=$ A validation-feasible and $\epsilon$-optimal package for $Q(S)$ \\ \makebox[17pt]{} (or NULL if unsolvable)

\State $\mathcal{O} \gets \textsc{GenerateScenarios}(m)$ \Comment{Initial optimization scenarios}
\State \textsf{qsSuccess}, $x^{\text{D}} \gets $\textsf{QuickSolve}(Q, S) \Comment{Check if solution is ``easy''}\label{li:quicksolve}
\If{\textsf{qsSuccess} $=$ True}\label{li:quickreturnIf}
\State \Return $x^{\text{D}}$ \Comment{Either solution to $Q(S)$ or NULL}\label{li:quickreturn}
\EndIf

\State $\alpha, V\gets \textsc{GetParams}(Q)$ \Comment{Extract $(\alpha_r,V_r)$ for each $r\in Q$}\label{li:getparams}
\State $\omega_0\gets \textsc{OmegaUpperBound}(Q,S,x^{\text{D}},\mathcal{V})$\label{li:EstObjVal} \label{li:getUB}
\While{True}
\State $V'\gets V$, $\alpha'\gets 1$ \Comment{Initial \LCVaR parameter values}
\State Compute $\hat V_0$ as in \eqref{eq:defVL}
and set $V_L\gets \hat V_0$, $V_U\gets V$, $\alpha_L\gets \alpha$, $\alpha_U\gets 1$
\While{True}
\State \textcolor{gray}{$\triangleright$ \textit{Search for $\alpha$}}
\State $\alpha'_\text{old}\gets \alpha'$
\State  status, $\alpha', x\gets \textsc{$\alpha$-Search}(V', \alpha_L,\alpha_U,Q,S,\epsilon,\delta,\mathcal{O},\mathcal{V},\omega_0)$\label{li:alphasearch}
\State \algorithmicif\ status.NeedScenarios $=$ True \algorithmicthen\ \textbf{break}\label{li:needScenariosalpha}
\State \algorithmicif\ status.Done $=$ True \algorithmicthen\ \textbf{return} $x$\label{li:donealpha}
\State \textcolor{gray}{$\triangleright$ \textit{Search for V}}
\State $V'_\text{old}\gets V'$
\State  status, $V', x\gets \textsc{V-Search}(\alpha', V_L,V_U,Q,S,\epsilon,\delta,\mathcal{O},\mathcal{V},\omega_0)$\label{li:Vsearch}
\State \algorithmicif\ status.NeedScenarios $=$ True \algorithmicthen\ \textbf{break} \label{li:needScenariosV}
\State \algorithmicif\ status.Done $=$ True \algorithmicthen\ \textbf{return} $x$\label{li:doneV}
\State $V_U\gets \max(V'-\delta,V_L)$
\State \algorithmicif\ $\max(V'_\text{old}-V',\alpha'_\text{old}-\alpha')<\delta$ \algorithmicthen\ \textbf{break} \label{li:noprogress}
\EndWhile
\State \algorithmicif\ $2m>|\mathcal{V}|$ \algorithmicthen\ \textbf{return} $x$ \Comment{Best solution found so far}\label{li:genscenariosA}
\State $\mathcal{O} \gets \mathcal{O} \cup \textsc{GenerateScenarios}(m)$\Comment{Double \# of scenarios}\label{li:genscenariosB} 
\State $m\gets 2m$\label{li:genscenariosC}
\EndWhile

\end{algorithmic}}
\end{algorithm}

\smallskip
\noindent
\textbf{Overview.} \looseness-1
\reva{The basic idea is to replace each \emph{risk constraint} (VaR or CVaR constraint) in an SPQ with a corresponding \LCVaR constraint of the form $\LCVaR_\alpha(x, A) \ge V$. Such replacement ensures that the resulting SAA problem size is independent of the number of scenarios. Moreover, the feasible region (in the space of possible packages), which was formerly non-linear and non-convex, is transformed to a convex polytope, so that the modified SAA problem can be efficiently solved using an ILP solver. We carefully choose the \LCVaR constraints so that (1)~the package solution for the re-shaped feasible region is \emph{validation-feasible}, i.e., satisfies the \emph{original} set of risk constraints with respect to the validation scenarios, and (2)~the package solution is \emph{$\epsilon$--optimal}, i.e., its objective value is $\ge (1-\epsilon)\omega$, where $\omega$ is the objective value for the true optimal solution to the original SPQ and $\epsilon$ is an application-specific error tolerance. Since the value of $\omega$ is unknown, we use an upper bound $\omegaU\ge \omega$, so that any package with objective value above $(1-\epsilon)\omegaU$ will have an objective value above $(1-\epsilon)\omega$ (see below). One simple choice for $\omegaU$ that works well in practice is $\omega_0$---the objective value of the package solution $x^{\text{D}}$ to query $Q^{\text{D}}(S)$, where $Q^{\text{D}}(S)$ is the deterministic package query obtained by removing all probabilistic constraints from $Q(S)$, i.e., removing all VaR, CVaR and expected sum constraints (line~\ref{li:getUB}).

\emph{Risk-constraint linearization} (RCL) is the process of replacing all risk constraints in a query $Q(S)$ by \LCVaR constraints to achieve objectives (1) and (2) above. RCL is accomplished via a search over potential $(\alpha,V)$ values for every \LCVaR constraint. For each choice of values, we solve the resulting SPQ via SAA approximation using the optimization scenarios and then check whether the returned package is validation-feasible and $\epsilon$--optimal (as conservatively estimated above).

To design an effective search strategy,  we first observe that, by Theorem~\ref{th:constraintCompare}, an \LCVaR constraint with parameters $\alpha$ and $V$ is more restrictive than a risk constraint with the same parameters. This strict \LCVaR parameterization, when applied to all constraints, results in a feasible-region polytope that is likely too small and may not contain the true optimal solution. Thus, the resulting package solution, while likely validation-feasible, will also likely be far from $\epsilon$--optimal. 
By varying the $\alpha$ and $V$ parameters, we can systematically shift and rotate the \LCVaR constraint boundaries until the feasible region contains a solution that is validation-feasible and $\epsilon$-optimal. As discussed below, for each constraint it suffices to search for optimal parameters in a bounded set of the form $[\alpha,1]\times[V_0,V]$, where $\alpha$ and $V$ are the parameters of the original VaR or CVaR constraint, which lets us use an efficient alternating-parameter search algorithm that navigates between solutions that are suboptimal and solutions that are validation-infeasible to find the desired package solution. The user will know whether the parameter search was successful or not. If successful, \lcsolve will return a validation-feasible and $\epsilon$--optimal package; otherwise, it will return the best validation-feasible package encountered so far, but with no optimality guarantees. Although the latter behavior is theoretically possible, we did not encounter any failed parameter searches in our experiments.}

\smallskip\noindent\textbf{Quick solution in special cases.}
Sometimes a query $Q(S)$ can be solved quickly, without going through the entire RCL process.
If the ILP solver can find a package solution $x^{\text{D}}$ to query $Q^{\text{D}}(S)$ as above, and if this solution is validation-feasible for $Q(S)$, then $x^{\text{D}}$ is the solution to $Q(S)$, since it satisfies all constraints, including the probabilistic ones, and the objective value is as high as possible since the probabilistic constraints have been completely relaxed, thereby maximizing the feasible region.
On the other hand, if no feasible solution can be found for $Q^{\text{D}}(S)$, then $Q(S)$ is also unsolvable, since $Q$ has all the constraints that $Q^{\text{D}}$ has, and more. We denote by \textsc{QuickSolve}$(Q,S)$ the subroutine that detects these two possible outcomes (line~\ref{li:quicksolve}). \textsc{QuickSolve}$(Q,S)$ returns a pair $(\textsf{qsSuccess},x^{\text{D}})$. If $\textsf{qsSuccess}= \textsf{True}$, then $x^{\text{D}}$ either is the package solution to $Q(S)$ or $x^{\text{D}}= \textsf{NULL}$, which indicates that $Q(S)$ is unsolvable, and \lcsolve terminates and returns one of these values (lines~\ref{li:quickreturnIf}--\ref{li:quickreturn}). If $\textsf{qsSuccess}= \textsf{False}$ then $x^{\text{D}}$ solves $Q^{\text{D}}(S)$ but is validation-infeasible for $Q(S)$ and the full RCL process is needed.




\smallskip\noindent\textbf{The case of a single risk constraint.} To understand how RCL works in the case that $\textsf{qsSuccess}= \textsf{False}$, first consider an SPQ $Q(S)$ having a single risk constraint---on a random attribute $A$ with parameters $(\alpha,V)$---that will be replaced by an \LCVaR constraint. If the \LCVaR constraint is too stringent, the package solution will be validation-feasible, but sub-optimal, whereas if the constraint is too relaxed, the package solution will be validation-infeasible. A package solution may also be validation-infeasible or sub-optimal if the number of optimization scenarios is too small so that the SAA approximation is not accurate. We will discuss how RCL deals with this latter issue shortly, but suppose for now that the initial number $m$ of optimization scenarios is adequate.

Let $\alpha'$ and $V'$ denote the adjusted parameters of the \LCVaR constraint. Recall \reva{from Theorem~\ref{th:constraintCompare}} that, for any risk constraint with parameters $\alpha$ and $V$, setting $\alpha'=\alpha$ and $V'=V$ will cause the corresponding \LCVaR constraint to be overly restrictive. Also note that an \LCVaR constraint of the form $\LCVaR_{\alpha'}(x, A) \ge V'$ becomes less restrictive as either $V'$ decreases or $\alpha'$ increases. The parameter $\alpha'$ can be set as high as $1$, in which case each coefficient $\hat C_i=\LCVaR_{\alpha'}(t_i.A)$ will be the average of the $t_i.A$ values across \emph{all} the optimization scenarios. On the other hand, for a given value of $\alpha$, the parameter $V'$ can be set as low as $V_0$, where $V_0=\sum_{i=1}^n \text{CVaR}_\alpha(t_i.A)*x^D_i$ and $x^D=(x^D_1,\ldots,x^D_n)$ is the package solution to the probabilistically unconstrained problem $Q^D(S)$. Note that, given a set $\mathcal{O}$ of optimization scenarios, we can approximate $V_0$ as previously discussed:
\begin{equation}\label{eq:defVL}
\hat V_0=\sum_{i=1}^n \widehat{\text{CVaR}}_\alpha(t_i.A,\mathcal{O})*x^D_i.
\end{equation}

\reva{For any ``non-trivial'' risk constraint, i.e., a constraint that is not satisfied by $x^D$, setting the \LCVaR parameters to $(V',\alpha')=(V_0,1)$ will necessarily yield a validation-infeasible package that does not satisfy the constraint.}
We thus need to search for the \reva{least restrictive} value of $V'$ and $\alpha'$ between the limits $[V_0, V]$ $[\alpha, 1]$ respectively \reva{that results in a validation-feasible and $\epsilon$--optimal package solution.}


\smallskip\noindent\textbf{Alternating parameter search (APS).} The RCL process starts by setting the \LCVaR constraint values to $(V',\alpha')=(V,1)$. Using the maximal value of $\alpha'$ will most likely make the package solution validation-infeasible. 
To attain feasibility, we make the \LCVaR constraint more restrictive via an \emph{$\alpha$-search}, that is, by decreasing $\alpha'$, using bisection, down to the maximum value for which the resulting package is validation-feasible, while keeping $V'$ unchanged (line~\ref{li:alphasearch})---decreasing $\alpha'$ further would lower the objective value.
After the above $\alpha$-search terminates with a validation-feasible package, we improve the objective value (while maintaining feasibility) by making the \LCVaR constraint less restrictive via a \emph{$V$-search}, that is, by decreasing $V'$, via bisection, down to the maximum value for which the corresponding package remains validation-feasible, while keeping $\alpha'$ unchanged (line~\ref{li:Vsearch}). Each bisection search for $V'$ terminates when the length of the search interval falls below a specified small constant $\delta$;
in our experiments we found that $\delta=10^{-3}$ was an effective and robust choice.

We then pivot back to adjusting $\alpha'$. Specifically, we decrease the new value of $V'$ by $\delta$, so that the corresponding package based on $(V'-\delta,\alpha')$ is now validation-infeasible. We then decrease $\alpha'$ using an $\alpha$-search to regain feasibility, then decrease $V'$ via a $V$-search to improve the objective, and so on. If any intermediate package is validation-feasible for the original SPQ $Q(S)$ and is near-optimal, i.e., the objective value equals or exceeds $(1-\epsilon)\omega_0$, then we immediately terminate the RCL search and return this package as our solution, setting the variable \textsf{Done} to \textsf{True} (lines~\ref{li:donealpha} and \ref{li:doneV}). The foregoing procedure works even if the initial $(V,1)$ \LCVaR constraint yields a validation-feasible package; in this case the first $\alpha$-search will trivially result in  $(V',\alpha')=(V,1)$, and the subsequent $V$-search will try to improve the objective value.

Our motivation for using APS is that it provably finds the optimal \LCVaR parameterization in the idealized setting where we can solve any SPQ exactly, determine the feasibility and objective value for any package exactly, and compute any VaR or CVaR exactly, without the need for SAA. Specifically, for an SPQ $Q(S)$ with one risk constraint parameterized by $(\alpha,V)$, denote by $\mathcal{F}_{Q(S)}$ the set of all \LCVaR parameterizations $(\alpha',V')\in [\alpha, 1]\times [V_0, V]$ for which the transformed SPQ $Q'(S;\alpha',V')$ with linearized risk constraint admits a feasible package solution. For any $(\alpha',V')\in\mathcal{F}_{Q(S)}$, denote by $\text{Obj}(\alpha',V')$ the objective value of its corresponding package solution. \reva{Intuitively, let $(\alpha^*,V^*)$ be the best parametrization in $\mathcal{F}_{Q(S)}$. Since APS minimally decreases $\alpha'$ such that the resulting package is validation-feasible (while maximizing the objective with the current value of $V'$), it is necessary that APS finds the point $(\alpha^*,\tilde{V})$ for some $\tilde{V}$. The following $V$-search will find $(\alpha^*,V^*)$}. Formally, we have the following result: 

\begin{theorem}\label{th:alternating_parameter_search}
Let $Q(S)$ be is a solvable SPQ with one risk constraint and suppose that $\mathcal{F}_{Q(S)}$ is nonempty. Then APS, under the idealized-setting assumption, will find a parameterization $(\alpha^*,V^*)\in\mathcal{F}_{Q(S)}$ such that $\text{Obj}(\alpha^*,V^*)\ge \text{Obj}(\alpha',V')$ for all $(\alpha',V')\in\mathcal{F}_{Q(S)}$.
\end{theorem}

In practice, we run APS using a set $\mathcal{O}$ of optimization scenarios to compute packages and validation scenarios $\mathcal{V}$ assess feasibility and objective values. The ideal setting essentially corresponds to $|\mathcal{O}|=|\mathcal {V}|=\infty$, so we expect APS to work increasingly well as $|\mathcal{O}|$ and $|\mathcal {V}|$ increase. Our experiments in Section~\ref{sec:experimental_evaluation} indicate good practical performance for the ranges of $|\mathcal{O}|$ and $|\mathcal {V}|$ that we consider.

\smallskip\noindent\textbf{Multiple risk constraints.} In general the set $R$ of risk constraints appearing in $Q$ satisfies $|R|>1$, so that quantities such as $V$ and $\alpha$ are vectors; the function \textsc{GetParams} parses the query $Q$, extracts the values $V_r$ and $\alpha_r$ for each constraint $r\in R$, and organizes the values into the vectors $V$ and $\alpha$ (line~\ref{li:getparams}). The \textsc{QuickSolve}$(Q,S)$ function works basically as described above. Assuming that qsSuccess $=$ False, so that the full RCL procedure is needed, the linearization process starts by setting $\alpha'_r=1$ and $V'=V_r$ for each $r\in R$, so that one or more of the risk constraints $r\in Q$ are validation-violated; denote by $R^*\subseteq R$ the set of problematic risk constraints. The search then decreases the problematic $\alpha'_r$ values via individual $\alpha$-searches until all the SPQ constraints are validation-satisfied. These $\alpha$-searches are synchronized in that, at each search step, a bisection operation is executed for each $\alpha'_r$ value with $r\in R^*$ and then the resulting overall set of risk constraints is checked for validation-satisfaction. The $\alpha$-search terminates when bisection has terminated for all of the problematic constraints due to convergence.\footnote{In theory, an initially non-problematic constraint might become problematic during the $V$-search, but we did not observe this behavior in any of our experiments. In any case, one can modify the algorithm slightly to deal with this rare situation.} 
When the $\alpha$-search terminates, the RCL procedure next decreases each $V'_r$ value by $\delta$ and commences a set of synchronized $V$-searches, and the two types of search alternate as before until a solution is produced.




\smallskip\noindent\textbf{Increasing the number of optimization scenarios.} \looseness-1
We have been assuming that the initial number $m$ of scenarios is sufficiently large so that the SAA ILP well approximates the true SPQ, but this may not be the case in general. Therefore, right after each (synchronized) bisection step during an $\alpha$-search or $V$-search, the RCL process examines the intermediate package corresponding to the new set of constraints just produced. If there is a significant relative difference between any VaR, CVaR, or expected value
when computed over the optimization scenarios and when computed over the validation scenarios,
then it follows that SAA approximation is poor.
If any relative difference is larger than $\epsilon$, then
the $\alpha$-search or $V$-search terminates immediately and sets \textsf{NeedScenarios} to \textsf{True} (lines~\ref{li:needScenariosalpha} and \ref{li:needScenariosV}). We then increase the number of scenarios (line~\ref{li:genscenariosA}) and restart the RCL process. We also increase the number of scenarios if the values of both $\alpha'$ and $V'$ stop changing and a solution package has not yet been found (line~\ref{li:noprogress}). If the number of optimization scenarios starts to exceed the number of validation scenarios, we terminate the algorithm and return the best solution so far (line~\ref{li:genscenariosA}).

\smallskip\noindent\textbf{Termination and guarantees.}
Assuming that the set $\mathcal{V}$ of validation scenarios is sufficiently large and that the upper bound $\omega_0$ is sufficiently tight, \lcsolve will terminate at line~\ref{li:quickreturn}, \ref{li:donealpha}, or \ref{li:doneV} with a near-optimal solution to $Q(S)$ or will terminate at line~\ref{li:quickreturn} with a \textsf{NULL} solution indicating that the problem is unsolvable. This is the behavior we observed in all of our experiments. Otherwise, \lcsolve will terminate at line~\ref{li:genscenariosB} and return the best package found so far, but with no guarantees on optimality. In this case, the user can increase the number of validation scenarios and try again. 

\section{\sskr}\label{section:sskr}

\begin{figure}[t]
    \centering
    \includegraphics[width=1\linewidth]{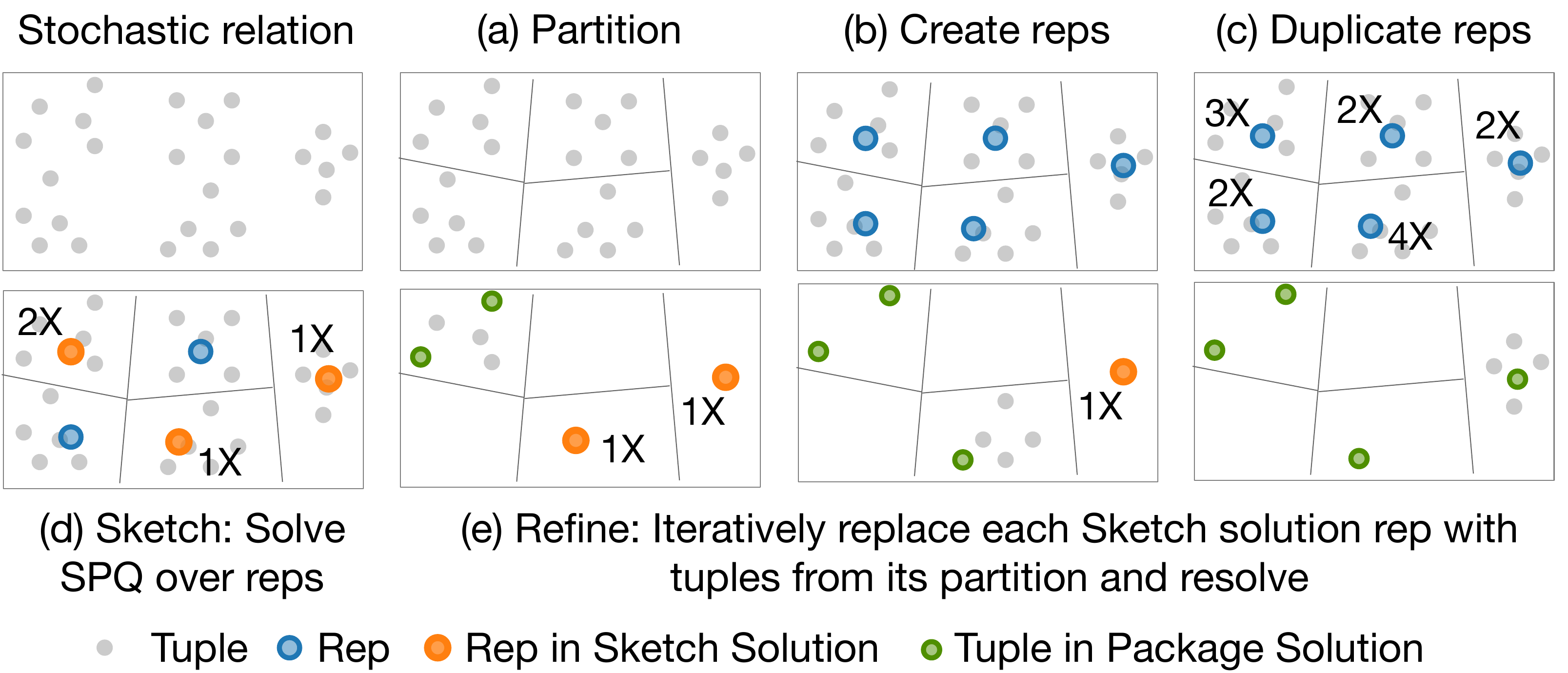}
    \vspace{-6mm}
    \caption{Solving a stochastic package query (SPQ) with \sskr.}
    \vspace{-4pt}
    \label{fig:stochastic-sketchrefine-hlf}
\end{figure}

\lcsolve alone cannot handle SPQs with very large numbers of tuples. To achieve scalability across the tuple dimension, we propose a new evaluation mechanism, \sskr. This algorithm follows the divide-and-conquer approach of \skr~\cite{brucato2018package} but makes the necessary (and nontrivial) modifications needed to handle uncertain data.

Figure~\ref{fig:stochastic-sketchrefine-hlf} gives a high-level illustration of the processing steps. The algorithm first partitions the data offline into groups of similar tuples 
(Figure~\ref{fig:stochastic-sketchrefine-hlf}a)
using \dpart, a novel trivially-par\-al\-lel\-iz\-a\-ble stochastic partitioning scheme (Section~\ref{sec:partitioning}) which ensures that (1)~tuples in the same partition are similar, 
and (2)~the maximum number $\tau$ of tuples allowed in each partition is small enough so that an ILP based on these tuples can be solved quickly.

\sskr then creates a \emph{representative} for each partition
(Figure~\ref{fig:stochastic-sketchrefine-hlf}b). A representative of a partition is an artificial tuple in which the value of each deterministic attribute is the average value over the tuples in the partition, and the stochastic-attribute distributions equal those of a selected tuple from the partition; Section~\ref{subsec:sketch} below details the selection process.

The next step is to create a \emph{sketch} $S_0$ using the representatives. In the deterministic setting, it suffices to simply define the sketch as the union of the representatives; the corresponding package solution specifies a multiplicity $\mult(t)\in [0,1,2,\ldots]$ for each representative $t$.
In the stochastic setting, however, this can lead to infeasible or sub-optimal packages due to inflation of risk---see Section~\ref{subsec:sketch} below---so instead of relying solely on multiplicities, we augment the sketch $S_0$ with distinct but stochastically identical \emph{duplicates} of each representative
(Figure~\ref{fig:stochastic-sketchrefine-hlf}c). To account for statistical correlation between the tuples within a partition, the duplicates must also be mutually correlated---a Gaussian copula method called NORTA is used to generate correlated duplicates~\cite{ghosh2003behavior}.

\sskr then solves $Q(S_0)$ using a slight variant of \lcsolve, called \lcsolve-S, 
thereby assigning a multiplicity to each duplicate (Figure~\ref{fig:stochastic-sketchrefine-hlf}d)
and creating a package solution. This package solution is then \emph{refined} into a new set of tuples $S_1$ by selecting a (subset of) partitions containing at least one tuple appearing in the package solution and replacing the 
duplicates in these partitions by the original tuples in them
(Figure~\ref{fig:stochastic-sketchrefine-hlf}e); because the size of each partition is bounded, the resulting query $Q(S_1)$ can be solved quickly using a slight variant of \lcsolve called \lcsolve-R. This refinement process continues
until \revb{all duplicates in the sketch solution are replaced by tuples from their respective partitions, thus obtaining the final package solution}. The following sections describe these algorithmic components in detail.

\subsection{The Sketch Phase} \label{subsec:sketch}

\begin{algorithm}[t]
\caption{\textsc{\solvesketch}}\label{alg:sketch}
{\footnotesize
\begin{algorithmic}[1]
\Require $Q:=$ A Stochastic Package Query
\Statex $T^\text{R} :=$ A set of representative tuples
\Statex $\tau :=$ Maximum number of tuples in a partition
\Statex $\Delta\Gamma :=$ Change in relative risk tolerance 
\Statex $m :=$ Initial number of optimization scenarios
\Statex $\mathcal{V} :=$ Set of validation scenarios
\Statex $P_\text{max} :=$ Max. number of distinct tuples in a package
\Statex $\epsilon :=$ Error approximation bound
\Statex $\delta :=$ Bisection termination distance

\Ensure $x :=$ A validation-feasible + $\epsilon$-optimal package for $Q(S_0)$ 

\State \textsf{qsSuccess}, $x^{\text{D}} \gets \textsf{QuickSolve}(Q, T^\text{R})$ \Comment{Is the solution ``easy''?}\label{li:quicksolveSSKRa}
\If{\textsf{qsSuccess} $=$ \textsf{True}}
\State \Return $x^{\text{D}}$ \Comment{Either sketch solution or NULL}\label{li:quicksolveSSKRb}
\EndIf
\State $\Gamma \gets$ \textsc{InitialRiskTolerance($\tau$)}\label{li:gammainit}
\While{$m \leq |\mathcal{V}|$}
      \State $S_0 \gets$ $\emptyset$ 
      \For{$t \in T^\text{R}$}\label{li:genloop}
       \State $d \gets \textsc{NumberOfDuplicates}(t, Q, \Gamma, m, P_\text{max})$\label{li:numdupes}
       \State $H_t\gets \textsc{Partition}(t)$ \Comment{Set of tuples represented by $t$}
       \State $\overline{\rho}\gets \textsc{MedianCorrelation}(t,H_t)$\label{li:mediancorr}
       \State $S_0\gets S_0\cup \textsc{GenDuplicates}(t,d,\overline{\rho})$\label{li:gendupes}
      \EndFor
      \State $x_0, \textsc{NeedsScenarios} \gets \lcsolveS(Q, S_0, m, \mathcal{V}, \delta,\epsilon)$\label{li:solvesketch}
      \If{$\textsc{NeedsScenarios} = \textsc{True}$}\label{li:immediatescenariosA}
      \State $m \gets \min(2m, |\mathcal{V}|)$\label{li:immediatescenariosB}
      \Else
      \If{$x_0\not=$ NULL} \textbf{\Return} $x_0$
      \Else
      \If{$\Gamma > 0.0$}
      \State $\Gamma \gets \max(\Gamma - \Delta\Gamma, 0.0)$\label{li:gammaDecr}
      \Else
      \State $m \gets \min(2m, |\mathcal{V}|)$~\label{li:reiterate}
      \EndIf
      \EndIf
      \EndIf
\EndWhile
\State \Return \textsc{NULL} \Comment{means the problem is unsolvable}
\end{algorithmic}}
\end{algorithm}

\textsc{SolveSketch} (Algorithm~\ref{alg:sketch}) creates a sketch $S_0$ using the set $T^\text{R}$ of representatives produced by \dpart and then solves the resulting query $Q(S_0)$ to create an initial sketch solution package $x_0$. As with \lcsolve, the \textsc{QuickSolve} function tries to solve the deterministic package query $Q^\text{D}(T^\text{R})$ obtained by removing all risk constraints from $Q$. If a solution $x^\text{D}$ exists that is validation-feasible with respect to $Q$, it is returned as the sketch solution $x_0$; if a solution to $Q^\text{D}(T^\text{R})$ does not exist, then \textsc{SolveSketch} declares the sketch problem unsolvable and returns \textsf{NULL} (lines~\ref{li:quicksolveSSKRa}--\ref{li:quicksolveSSKRb}).

If $\textsf{qsSuccess}=\textsc{False}$, then the full query $Q$ must be solved. The key difference from the deterministic setting is the need to create duplicates of each representative in $T^\text{R}$ to form the initial sketch $S_0$. The key issues are (i)~determining the number of duplicates to create for each representative, (ii)~determining the appropriate correlation between a set of duplicates, and (iii)~generating scenarios from correlated duplicates. We first motivate the use of duplicates and then discuss how issues~(i) and (ii) are addressed; see 
\appref{appendix:NORTA}
for details on how the NORTA procedure is used to generate scenarios from correlated duplicates.

\smallskip\noindent\textbf{The need for duplicates.} \revb{Duplicates are used in a partition $P$ to ensure that optimal tuples in $P$ do not get erroneously eliminated from consideration due to overestimation of risk. As a simple example, consider a  partition containing tuples $t_1,\ldots,t_k$ having a stochastic attribute $A$ such that $t_i.A$ has an independent $N(0,1)$ distribution (normal with mean~0 and variance~1) for all $t_i\in P$. Also suppose that the SPQ has a constraint of the form \lstinline[mathescape]!sum($A$) <= -2.0 with probability <= 0.1!. 
Suppose that the truly optimal package contains two tuples $t_1$ and $t_2$ from $P$, each with a multiplicity of 1. The sum $t_1.A+t_2.A$ has a $N(0,2)$ distribution, so that $\Prob(t_1.A+t_2.A\le -2.0)\approx 0.08$ and there are no feasibility issues. Here, losses from one tuple can be offset by gains in the other tuple, reducing the overall risk of high losses. However, if there is only a single representative $t_p\in P$ in the sketch of $P$, then a candidate sketch package that tried to use two tuples from $P$ (as in the optimal package) would end up using $t_p$ with multiplicity 2, so that $2t_p.A\sim N(0,4)$ and $\Prob(2t_p.A\le -2.0)\approx 0.16$, violating the constraint. Due to this overestimation of risk, the representative $t_p$ can have a multiplicity of at most 1 in the sketch package. Suppose that another representative $t_r$ from a different partition is also in the sketch package (a typical outcome). Then we cannot refine $[t_p:1,t_r:m]$ to the intermediate package $[t_1:1, t_2:1, t_r: m]$ for any $m\ge 1$ since this package is necessarily infeasible; if it were feasible, then $[t_1:1,t_2:1]$ would not be optimal since the objective of $[t_1:1, t_2:1, t_r: m]$ would be greater. Thus, the final package would have  at most one of $t_1$ and $t_2$, and hence be suboptimal.
We solve this problem by maintaining duplicates $t_p^{(1)}$ and $t_p^{(2)}$. Then, as with the actual tuples, a sketch solution could contain these duplicates, and not violate the constraint. 
After $P$ is refined, the solver would have the option of using both tuples $t_1$ and $t_2$ in the final package solution. In general, tuples within a partition might be correlated; in this case, the duplicates would also be correlated to accurately approximate the risk of using tuples in the partition when forming a package. We emphasize that, because of the refinement step, actual tuples, and not duplicates, appear in the final package solution. 
We provide an ablation study that empirically demonstrates the deleterious effect of not using duplicates~\appref{sec:ablation_dup}.}

\smallskip\noindent\textbf{Choosing the number of duplicates.} We assume that we have an upper bound $\Pmax$ on the number of distinct tuples in a package. This can often be derived 
from query constraints such as \lstinline[mathescape]!COUNT(*) $\leq \Pmax$!, or \lstinline[mathescape]!SUM(A) $\leq V$!,
or on domain knowledge.
In an extreme case, all of the $\Pmax$ tuples in the optimal package might belong to the same partition; if each representative has $\Pmax$ duplicates and if the optimal tuples are close to the representative, the sketch will yield a validation-feasible package. However, since each duplicate induces a corresponding decision variable, solver runtimes can be expensive.

We therefore aim to use $d_t<\Pmax$ duplicates for each representative $t$. Doing so incurs some risk: if the optimal package contains $\Pmax$ tuples and all of these tuples come from $t$'s partition, then the average multiplicity of each duplicate in the sketch solution is $\Pmax/d_t$. As we have seen, assigning a multiplicity $\ge 2$ to any duplicate results in an increased probability of each risk constraint $r\in R$ being violated. Specifically, for a given constraint $r$ that imposes a lower bound on a risk metric (VaR or CVaR), the risk value of the package based on $d_t$ duplicates in the sketch will be lower than that of the package formed using $\Pmax$ duplicates and hence more likely to violate the lower bound. We measure the relative increase in risk with respect to $r$ due to using $d_t$ duplicates instead of $\Pmax$ duplicates by $\gamma(r,d_t)=\bigl(\text{Risk}_r(S_\text{max}) - \text{Risk}_r(S_{d_t})\bigr)/|\text{Risk}_r(S_\text{max})|$, 
where $S_\text{max}$ and $S_{d_t}$ are sketches consisting of $\Pmax$ and $d_t$ duplicates of $t$, and $\text{Risk}_r$ is the VaR or CVaR function that appears in constraint~$r$. Note that $\gamma(r,d)$ increases with decreasing $d$. We determine the number of duplicates $d_t(r)$ as the smallest $d$ such that $\gamma(r,d) \le \Gamma$, where $\Gamma$ is a specified \emph{risk tolerance ratio}. We then set $d_t=\min(\max_{r\in R} d_t(r), |P_t|)$, where $|P_t|$ is the number of tuples in the $t$'s partition. \revb{(Thus the number of duplicates is upper-bounded by the number of tuples in $P_t$.)}
The function \textsc{NumberOfDuplicates} in line~\ref{li:numdupes} performs these calculations.

Note that, as $\Gamma$ decreases, the number of required duplicates $d$ increases. Using bisection, we initially choose $\Gamma\in[0,1]$ such that $\sum_{t \in T^\text{R}}d_t \le \max(\tau, |T^\text{R}|)$, where $T^\text{R}$ is the set of representatives and $\tau$ is an upper bound on the 
total number of duplicates (line~\ref{li:gammainit}).
The bound $\tau$ is set such that in-memory ILP solvers like \textsc{Gurobi} can solve problems with $\tau$ tuples within an acceptable amount of time. The number of tolerable decision variables can vary due to underlying hardware attributes of the system, differences in the solver software, runtime versus quality requirements, and so on. 

\smallskip\noindent\textbf{Choosing the correlation between duplicates.} 
We have defined duplicates to be distinct but stochastically identical, by which we mean that for a set $t^{(1)},\ldots,t^{(d)}$ of duplicates and each stochastic attribute $A$, the marginal distributions of $t^{(1)}.A, \ldots,t^{(d)}.A$ are the same. To motivate the notion of duplicates we gave a small example using two statistically independent duplicates having a common $N(0,1)$ marginal distribution. In general, however, assuming mutual independence among duplicates ignores correlations between the actual tuples within a partition, which can cause problems in the refine phase. E.g., if all tuples in the partition are highly correlated, then sums of the form $\sum_i t_i.A*x_i$ are subject to more severe fluctuations than sums involving independent $t_i.A$'s, and the risk of violating VaR or CVaR constraints is higher. Thus, a sketch package with independent duplicates for a given partition may underestimate the risks of including actual tuples from that partition, and no feasible solution may be found during refine, when the solver tries to replace independent duplicates with positively-correlated tuples. Duplicates of each representative should therefore roughly mirror the correlation between tuples in the partition.

We use the median $\overline{\rho}$ of the pairwise Pearson correlation coefficients between the representative and all other tuples in its partition as the correlation coefficient between each pair of duplicates (line~\ref{li:mediancorr}). 
Using the median ensures that the number of tuples having a higher correlation with the representative relative to the duplicates equals the number of tuples having a lower correlation with the representative. Thus, if one of the duplicates is replaced by a tuple having a higher correlation coefficient than the median, thereby increasing risks compared to the sketch package, the solver will have the chance to balance it out by taking another similarly optimal tuple that has a lower correlation coefficient. (We actually take the maximum of $\overline{\rho}$ and 0 so that the correlation matrix for the duplicates is positive semi-definite. Our clustering method---see Section~\ref{sec:partitioning}--- ensures that that the median value is virtually always nonnegative so that no correction is needed.)

\smallskip\noindent\textbf{Computing the sketch package.} \looseness-1
As discussed, \textsc{SolveSketch} initially sets $\gamma$ in line~\ref{li:gammainit} so as to ensure at most $\tau$ duplicates
overall. In lines~\ref{li:genloop}--\ref{li:gendupes} the initial set of duplicates are created to form an initial version of the sketch $S_0$. Then a slight variant of \lcsolve is used to try and solve $Q(S_0)$ (line~\ref{li:solvesketch}). This variant, \lcsolveS, is almost identical to \lcsolve, except that the process of increasing the number of scenarios is now controlled by the calling function \textsc{SolveSketch}. Specifically, if $\textsf{NeedScenarios}=\textsf{True}$ as in lines~\ref{li:needScenariosalpha} or \ref{li:needScenariosV} of \lcsolve, then a \textsf{NeedScenarios} indicator variable is set to \textsf{True} and returned to \textsc{SolveSketch}. Moreover, lines~\ref{li:genscenariosA}--\ref{li:genscenariosC} in \lcsolve are replaced by a simple ``\textbf{return} \textsf{NULL}'' statement. Thus if the need for more scenarios is detected during an $\alpha$-search or $V$-search within \lcsolveS, this information is immediately sent to \textsc{SolveSketch}, which then increases the number of scenarios (lines~\ref{li:immediatescenariosA} and \ref{li:immediatescenariosB}). If \lcsolveS returns \textsf{NULL}, i.e., if the number of optimization scenarios appears adequate but an $\epsilon$-optimal and validation-feasible solution cannot be found while attempting to solve the sketch problem, then $\Gamma$ is decreased (line~\ref{li:gammaDecr}) and the resulting, larger sketch is tried. If $\Gamma = 0$ so that a decrease is impossible, we double the number of optimization scenarios (line~\ref{li:reiterate}) and try again, keeping $\Gamma=0$ since we know that we will need a lot of duplicates. Note that we try decreasing $\Gamma$---thereby adding duplicates---before we increase the number of scenarios because generating more optimization scenarios can incur significant computational costs, and is hence deferred as long as possible.

\subsection{The Refine Phase} \label{subsec:refine}

The refinement process is very similar to that of the deterministic \skr algorithm, so we give a brief overview and refer to reader to \cite{brucato2018package} for further details. The process iteratively replaces each synthetic representative selected in the sketch package with actual tuples from its own partition.
First, using the `Best Fit Decreasing' algorithm for bin-packing~\cite{garey1979computers}, the \textsc{Refine} algorithm bins the partitions with tuples in the sketch package into a near-minimum number of `partition groups' such that total number of tuples in each partition group is at most $\tau$. Then, in each step, it replaces all the selected duplicates from one of the partition groups, while preserving both the selected duplicates from the as yet unrefined partitions and the tuples selected during previous refinements. 

\smallskip\noindent\textbf{Number of optimization scenarios.} \textsc{Refine} uses as many optimization scenarios as were used to derive the sketch package. The intuition is that if $m$ optimization scenarios were enough to create a satisfactory package from the representative duplicates, they should suffice for doing the same from the actual tuples, since each tuple should be similar to their representatives.

\smallskip\noindent\textbf{Refinement order.}
Starting from a sketch package with tuples in $k$ distinct partitions that are binned into $b$ groups, there are $b!$ possible orders in which the partitions can be refined. Using an `incorrect' order can lead to an infeasible intermediate ILP, or one whose package solution is far less optimal than the sketch package. \textsc{Refine} attempts to select a correct order using a ``greedy backtracking'' technique as in~\cite{brucato2018package}. The idea is to start with a randomly selected permutation of partition groups, i.e., refinement order. If an intermediate refinement fails to find a validation-feasible 
package whose objective value is within $(1\pm\epsilon)$ of
the objective value of the sketch package, we ``undo'' the previous refinement and greedily attempt to refine the failing partition group in its place. We keep undoing the previous refinements until the partition group can be successfully refined. If the group cannot be refined even when it is moved back to the first position in the refinement order, then the stochastic behaviour of the tuples within the partition group is not adequately captured by the behavior of the duplicates for the partition. This can happen if the correlations among the duplicates are too small (Section~\ref{subsec:sketch}), so we increase the common correlation coefficient for each of those partitions from $\overline{\rho}$ to $\overline{\rho}+\Delta_{\rho}$ (we take $\Delta_{\rho}=0.1$). This discourages the sketch solver from taking more tuples from these partition groups. Afterwards, we re-execute the sketch query, and refine the newly obtained sketch package. Greedy backtracking continues to explore different orders in which the groups may be refined until an acceptable final package is found.

\smallskip\noindent\textbf{Refinement operation.}
At each step, \textsc{Refine} selects an unrefined partition group with at least one duplicate appearing in the sketch package. The refinement operation involves solving an ILP corresponding to $Q(S_c\cup S_p\cup S_u)$, where $Q$ is the SPQ of interest, $S_c$ comprises all of the actual tuples for the partition group currently being refined, $S_p$ comprises package tuples remaining from previously-refined groups, and $S_u$ comprises the duplicates from all as yet unrefined groups. The constraints of the ILP include all constraints in $Q$ (including the linearized risk constraints), as well as additional constraints ensuring that the multiplicities of the tuples in $S_p$ and $S_u$ remain unchanged. Thus the only change to the package is to replace the duplicates for the current group with zero or more actual tuples from its partitions. We use the variant \lcsolve-R to formulate and solve the ILP. The only differences between \lcsolve and \lcsolve-R are that (i)~in line~\ref{li:getUB} the upper-bound constant $\omega_0$ is instead chosen as the objective value $\Bar{\omega}$ corresponding to the package solution of the sketch problem $Q(S_0)$ produced by \solvesketch, (ii)~in line~\ref{li:genscenariosA}, NULL is returned rather than the best feasible solution so far and (iii)~the number of optimization scenarios is not increased. Thus a non-NULL package returned by \lcsolve-R will be validation-feasible with an objective value $\omega$ that satisfies $\omega\ge (1-\epsilon)\Bar{\omega}$.

\section{Stochastic partitioning}\label{sec:partitioning}



\looseness-1
\sskr needs tuples in a relation to be partitioned into sufficiently small groups of similar tuples prior to query execution. Specifically, 
each refine ILP has at least as many decision variables as tuples in the partition being refined. We therefore constrain the number of tuples in every partition to a given size threshold $\tau$.
The parameter $\tau$ is as in Algorithm~\ref{alg:sketch} and is chosen 
to ensure reasonable ILP-solver solution times.
Also, \sskr requires high inter-tuple similarity within partitions. During sketch, every tuple in a partition is represented by a single representative, and a representative that is not sufficiently similar to all the tuples in its partition may result in sketch packages that cannot be refined to feasible and near-optimal packages. 
Finding a suitable representative can be hard if tuples within a partition are too dissimilar with respect to their attributes.

\smallskip
\noindent
\textbf{Prior Clustering Methods.} Existing uncertain data clustering algorithms that cluster similar stochastic tuples together do not ensure the resulting clusters will satisfy size thresholds~\cite{chau2006uncertain, gullo2008clustering, gullo2012uncertain, gullo2013minimizing}. Although some hierarchical partitioning approaches~\cite{gullo2008hierarchical, zhang2017novel} can be modified to repeatedly repartition the clusters until no cluster has more than $\tau$ tuples, their runtime complexities are super-quadratic with respect to the number of tuples in the relation, making them unsuitable for fast partitioning of large relations. We therefore introduce \dpart, a partitioning algorithm having sub-quadratic time complexity and ensuring that no partition in a stochastic relation has more than $\tau$ tuples. We provide a high-level overview of the working principles of \dpart here, and refer interested readers to 
\appref{appendix:dpart_details}
for more details.

\smallskip
\noindent
\textbf{Diameter Thresholds.} To ensure that tuples within a partition are sufficiently similar, \dpart ensures that the ``distance'' (as defined below) between any pair of tuples in a partition with respect to each attribute $A$ is upper-bounded by a \emph{diameter threshold} $d_{A}$.
Values of $d_{A}$ for every attribute $A$ can be chosen by the user
based on runtime requirements.
Smaller values of $d_{A}$ produce tighter partitions whose representatives better reflect the stochastic properties and deterministic values of other tuples in their partitions, leading to better quality solutions, 
but also increase the number of partitions, thereby increasing the number of decision variables in the sketch ILP and consequently increasing runtime. In our experiments, we chose values of $d_{A}$ that keep the total number of partitions to approximately within $[\frac{\tau}{10}, \frac{\tau}{2}]$, to allow some wiggle room for Sketch to create duplicates without exceeding the size threshold $\tau$ or incurring excessive runtimes.

\smallskip
\noindent
\textbf{Inter-tuple distances.} \looseness-1
We calculate the distance between any two tuples $t_1$ and $t_2$ w.r.t.\ an attribute $A$ using their Mean Absolute Distance (\textsc{MAD}), which we define as $\textsc{MAD}(t_1.A, t_2.A)=\mathbb{E}[|t_1.A-t_2.A|]$. Thus, if $A$ is deterministic, then $\textsc{MAD}(t_1.A, t_2.A)=|t_1.A-t_2.A|$; if $A$ is stochastic, then we estimate MAD as the average of $|t_1.A-t_2.A|$ over a set of i.i.d.\ scenarios. \reva{We assume throughout that $\mathbb{E}[|t.A|]<\infty$ for all tuples $t$ and attributes $A$, so that the MAD is always well defined and finite.} Unlike other metrics such as Wasserstein distance~\cite{panaretos2019statistical} and KL-divergence~\cite{shlens2014notes}, which only consider similarities between probability density functions, MAD inherently accounts for correlations between tuples. Given $t_1$ and $t_2$ with similar PDFs for attribute $A$, the distance $\textsc{MAD}(t_1.A,t_2.A)$ is smaller when $t_1$ and $t_2$ are positively correlated than when they are independent; the case study in \appref{appendix:mad_case_study}
gives a concrete example of this phenomenon.
If inter-tuple MADs within a partition are small, then the tuples in the partition will have similar values across scenarios, allowing a representative to closely reflect their stochastic behavior. \reva{This property is formally established via Theorem~\ref{th:mad_cvar}, which shows that bounding the MAD between two tuples also bounds the difference in their CVaRs at the tails of their distributions.   
\begin{theorem}
Suppose the MAD between tuples $t_1$ and $t_2$ w.r.t. an attribute $C$ is bounded by $d_C$, i.e., $\mathbb{E}[|t_1.C-t_2.C|] \le d_C$. Then, the difference between the CVaRs of $t_1$ and $t_2$ in their lower $\alpha$-tail is bounded: $|\text{CVaR}_{\alpha}(t_1.C) - \text{CVaR}_{\alpha}(t_2.C)| \le \frac{d_C}{\alpha}$
\label{th:mad_cvar}
\end{theorem}

We give a formal proof in \appref{appendix:APSproof}. Intuitively, consider what happens if the theorem's conclusion is false
in an SAA setting: If the average value of $|t_1.C-t_2.C|$ exceeds $\frac{d_C}{\alpha}$ in any set of $\lfloor\alpha m \rfloor$ lowest-valued scenarios, 
then even if there is no (absolute) difference between their values in the remaining scenarios, $\mathbb{E}[|t_1.C-t_2.C|] > \alpha\frac{d_C}{\alpha} = d_C$, thus contradicting the hypothesis}. 
For efficiency in estimating MAD, \dpart generates a fixed set of scenarios before partitioning to avoid repeatedly generating scenarios on the fly. Generating more scenarios makes the estimation of MAD more accurate for stochastic attributes; in our experiments, we precomputed a set of $200$ scenarios. See \appref{appendix:mad_estimation}
for insight on how many scenarios are needed to obtain statistically accurate estimates of MAD.
\revb{Note that MAD is only defined w.r.t. a single attribute. We require diameter thresholds to be defined for each attribute separately, and the MAD between any two tuples must be lower than the diameter thresholds for all attributes in order for them to be in the same partition. Crucially, use of MAD allows us to formally guarantee that \sskr achieves a $(1-\epsilon)^2$-optimal solution; see \appref{sec:theory}.}

\smallskip
\noindent
\textbf{Triggering the partitioning of a set.} \dpart recursively partitions sets of stochastic tuples until no size or diameter constraints are violated. The size constraint is violated if a set has more than $\tau$ tuples. Exactly determining the diameter of a partition for an attribute $A$ would require estimating the MAD between every pair of tuples in the set.
To avoid this quadratic complexity, we use a conservative approach which exploits 
the fact that MAD, because it is based on the $L_1$ norm, inherits the triangle inequality: $\textsc{MAD}(t_1.A, t_3.A) \le \textsc{MAD}(t_1.A, t_2.A) + \textsc{MAD}(t_2.A, t_3.A)$. In an operation we call \pvscan, we find the distance of every other tuple in the set from a randomly chosen tuple $t$. Hence, if the distance of the farthest tuple $\dot t_{A}$ from $t$ is $\dot d_A$, the distance between any two tuples in the partition is bounded by $2\dot d_A$. \revb{If $2\dot d_A \le d_A$ is true for each attribute $A$, all the diameter constraints are necessarily satisfied}. If any constraint is violated, further partitioning is triggered.

\smallskip
\noindent
\textbf{\pvscan-based partitioning.} 
To partition further, \dpart (1)~executes a \pvscan for each attribute from a random tuple $t$ to identify the attribute $A^*$ with the highest diameter-to-threshold ratio: $A^*=\argmax_{A}2\dot d_A/d_A$, 
(2)~runs a second scan over all tuples in the set, but this time calculates their distances from the farthest tuple $\dot t_{A^*}$ found in the first scan, and (3)~stores a list of tuple IDs in increasing order of distance from $\dot t_{A^*}$. Further partitioning is done by segmenting the list and creating a sub-partition out of each segment. Specifically, 
if the size constraint is violated, then we partition the list into contiguous segments each containing $\le \tau$ tuples. Each resulting segment thus satisfies the size constraint.
If a diameter constraint is violated by any resulting segment (according to the conservative test described previously), it is recursively partitioned using distance-based partitioning. Let $\ddot d_{A^*}$ be the distance of the farthest tuple in the sub-partition from $\dot t_{A^*}$. For distance-based partitioning, we create $\lceil \ddot d_{A^*}/d_{A^*} \rceil$ sub-partitions where the $i$-th sub-partition contains all tuples within distance $[(i-1)\cdot d_{A^*}, i\cdot d_{A^*}]$. Each sub-partition is then recursively partitioned until no diameter constraints are violated. \pvscan is embarassingly parallel, and, after the initial size-based partitioning, diameter-based partitioning can be conducted in parallel for each partition. \dpart can hence be accelerated with multi-core processing.

\smallskip
\noindent
\textbf{Representative Selection.} After partitioning, we select a representative tuple for each partition. The value of each deterministic attribute of the representative tuple is equal to the mean of that attribute among every tuple in the partition. For stochastic attributes, naively computing a ``mean distribution'', e.g. as a weighted mixture, is computationally expensive and would require excessive storage.
Thus the distribution of each stochastic attribute of the representative is made equal to that of a chosen tuple in the partition, namely the tuple with the lowest mean distance to every other tuple. 

Naively searching for the tuple with the lowest mean distance would require quadratic MAD estimations.  Instead, we propose the total \emph{worst-case replacement} cost heuristic for selecting representatives. First, we compute the minimum and maximum values of an attribute $A$ over all the tuples in each scenario. In a given scenario, we define the worst-case \emph{replacement cost} of $t.A$ for a tuple $t$ as the greater of its absolute difference with that scenario's minimum or maximum. The cost for $t.A$ is obtained by summing its replacement costs over all the scenarios. The representative variable for $A$ is then chosen as that of the tuple with the minimum total cost. We provide pseudocodes for the \dpart and representative selection routines in \appref{appendix:dpart_details} and \appref{appendix:rep_details}.  

This representative tuple selection scheme assumes that values for different stochastic attributes are generated independently. See \appref{appendix:correlated_attributes} for a description of the minor modifications to representative selection and scenario generation 
when the stochastic attributes within a tuple can be statistically correlated.
\section{Experimental Evaluation} \label{sec:experimental_evaluation}

We show that (1)~\lcsolve produces packages of comparable quality as \ssearch in an order of magnitude lower runtime and (2)~\sskr scales to significantly larger data sizes than both of these approaches. \revc{Additional experiments are given in \appref{appendix:addlExpts}, including ablation studies justifying the use of duplicates and establishing the  superiority of \dpart relative to other partitioning schemes, as well as an experiment verifying the robustness of our package solutions.}
\vspace{-0.5em}
\subsection{Experimental Setup}

\textbf{Environment.} 
All approaches are implemented in Python 3.12.14. We use Postgres 16.1 as the supporting DBMS, and Gurobi 11.0.3 as the base solver.
We ran our experiments on a 2.66 GHz 16-core processor with 16 GB of RAM.
For each experiment, we report average results and standard deviation error bars over 16 runs.

\smallskip
\noindent
\textbf{Datasets.} 
We use two datasets in our experiments: (1)~We construct a stock investments table (Figure~\ref{fig:portfolio-builder}) using historical NASDAQ, NYSE, and S\&P-500 data~\cite{portfolio} with up to 4.8M tuples. We model gains using Geometric Brownian Motion with drift and volatility estimated from historical stock prices for 3289 companies. We vary holding periods from a half a day to 730 days.
(2)~We generate up to 6M LineItem tuples from the TPC-H V3 benchmark~\cite{poess2000new}.
We add Gaussian noise ($\mu_{\text{noise}} \sim \mathcal{N}(0, 1)$, $\sigma^2_{\text{noise}} \sim \text{Exp}(2)$) to price and quantity to introduce stochasticity.
Further details are in \appref{appendix:datasets}.


\smallskip
\noindent
\textbf{Workloads.} 
We generate a workload of stochastic package queries following the methodology of~\cite{mai2023scaling}: We vary the constraint bounds of the SPQs to generate queries of varying \emph{hardness} (H), measured as the negative log likelihood of the probability that a random package is feasible for that query; see \appref{appendix:hardness}.
For stocks, the queries find portfolios that maximize expected total gains, such that total price is below a fixed budget and total gains exceed a minimum amount with a given high probability. 
For TPC-H, the queries find packages of items that maximize expected total prices, such that total price exceeds a minimum bound with a given high probability, total quantity shipped is below a maximum amount with a given high probability, and total taxes are below a fixed amount. All packages have a size constraint of at most 30 tuples. 
We provide complete workload details in \appref{appendix:hardness}.

\begin{figure*}[t]
    \centering
    \includegraphics[width=\linewidth]{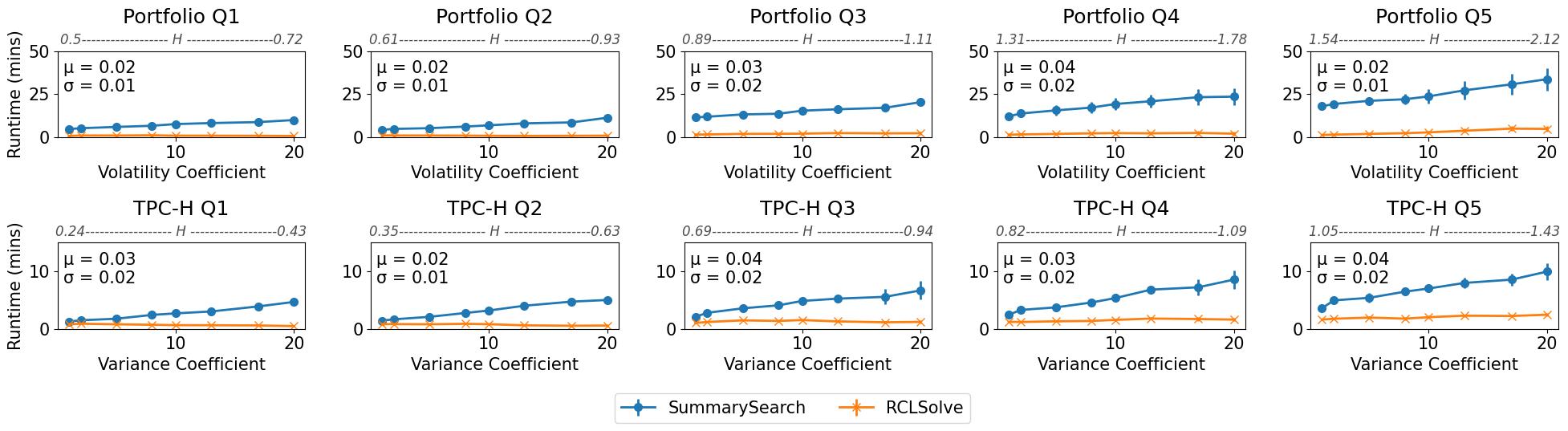}
    \vspace{-7mm}
    \caption{Increasing variance or volatility coefficients increases tuple uncertainty as well as query hardness. ($H$ range reported with each plot).  \lcsolve is faster than \ssearch especially at high variances.  
    In each plot, $\mu$ and $\sigma$ report the mean and standard deviation, respectively, of the relative integrality gap for \lcsolve's packages.
    }
    \label{fig:inc_vars}
\end{figure*}

\smallskip
\noindent
\textbf{Hyperparameters.} 
We set the number of initial optimization ($m)$ and validation scenarios ($\hat{m}$) to $100$ and $10^6$, respectively, for both \ssearch and \lcsolve. We use $\epsilon = 0.05$ as the approximation error bound, and set $\delta=10^{-2}$ as the bisection termination threshold for \lcsolve.
For \sskr, we set $P_\text{max}$ to 30, and during Sketch, we decrease the relative risk tolerance by $0.03$ in each iteration, i.e., $\Delta\Gamma = 0.03$.
%
%
We set the partitioning size threshold $\tau=10^5$, 
as Gurobi solves randomly-generated, satisfiable ILPs with 3 constraints in roughly one minute at this scale.
We set the diameter thresholds for the Portfolio dataset as $d_{\text{price}} =10$, $d_{\text{gain}} = 100$, and for TPC-H as $d_{\text{price}}=50$, $d_{\text{quantity}}=5$, $d_{\text{tax}}=0.05$. \revc{\appref{sec:hyper} contains details on how our results are not sensitive to small changes in parameter values and provides guidance on finding appropriate hyperparameter settings for different datasets.}

\smallskip
\noindent
\textbf{Metrics.} We report wall-clock \textit{query run times} and the \textit{relative integrality gap} $(\omega-\omega^*)/\omega^*$, where $\omega$ is objective value of the returned package and $\omega^*$ is that of the best package found on a relaxation of the query with integrality constraints removed. The latter ``best'' package is one found by either \lcsolve, \ssearch, or (on data sets with fewer than 40K tuples) na\"ive~\cite{BrucatoYAHM21}. We report the means ($\mu$) and standard deviations ($\sigma$) of the relative integrality gaps of packages formed by \sskr (on larger relations) and \lcsolve (on data sets with fewer than 40k tuples). 


\subsection{Main Results: Scalability and Optimality}

\textbf{Increasing uncertainty.}
\sloppy We evaluated our novel \lcsolve method against the state of the art in stochastic package query evaluation, \ssearch.  As we noted in Section~\ref{sec:intro}, \ssearch does not work well with high-variance stochastic attributes.  In Figure~\ref{fig:inc_vars}, we demonstrate the performance of the two methods as the uncertainty in the stochastic attributes increases, while keeping the datasets small (40K tuples for Portfolio and 20K for TPC-H). We control uncertainty by modifying the volatility  of the GBM model of a stock's gain, and the variance of Gaussian noise added to price and quantity in TPC-H.  Increasing uncertainty also increases query hardness, as indicated by the hardness ranges reported in each plot. \lcsolve is significantly faster than \ssearch and this difference is more pronounced with high uncertainty and increased hardness.  The reason is that \lcsolve generates fewer scenarios than \ssearch (approximately 6x fewer scenarios on average for the highest variance / volatility settings).  Moreover, with many scenarios, \ssearch not only solves more, but also harder optimization problems, that Gurobi takes minutes to solve.


The runtime gains of \lcsolve do not come at the cost of quality.  The packages it produced have a relative integrality gap in the 0.94--0.99 range, same as \ssearch.


\smallskip
\noindent
\fbox{
\parbox{0.96\columnwidth}{
\emph{Key takeaway:} 
   Our novel risk linearization approach is superior to the state of the art, by maintaining robust quality and fast performance in cases of increased data uncertainty.
}}

\begin{figure*}[t]
    \centering
    \includegraphics[width=\linewidth]{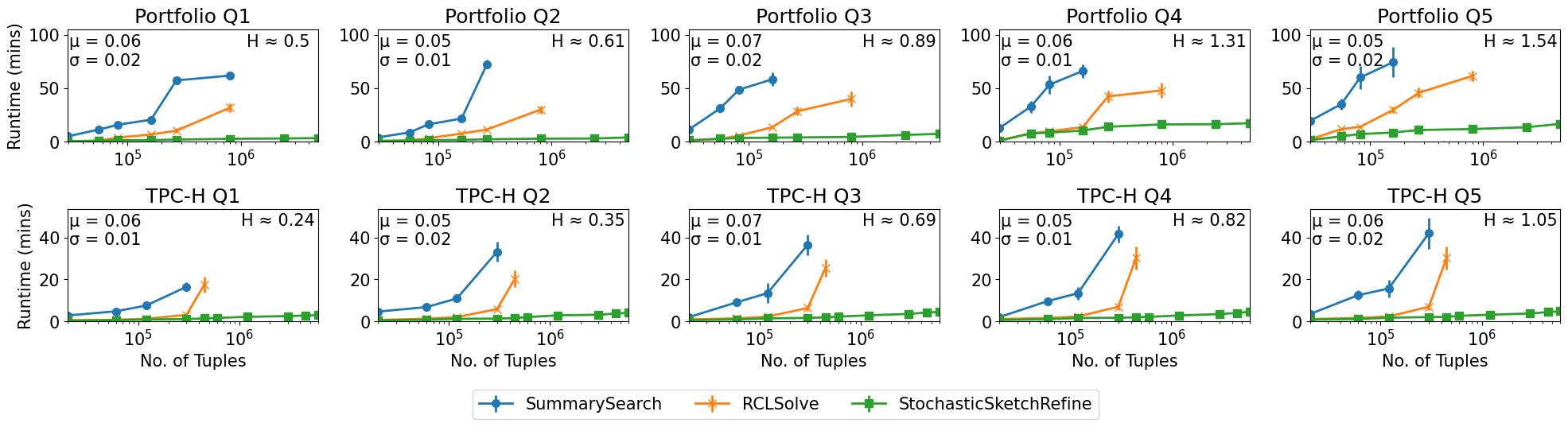}
    \vspace{-7mm}
    \caption{\lcsolve continues to outperform \ssearch but fails to scale beyond 1M tuples. The absence of data points indicates that no solutions were found within 1.5 hours. \sskr scales well as data size and query hardness ($H$) increase.
    Each plot shows the relative integrality gap statistics ($\mu, \sigma$) for \sskr's packages.
    }
    \label{fig:inc_tuples}
\end{figure*}

\smallskip
\noindent
\textbf{Increasing data size.} 
While we saw \lcsolve outperform \ssearch, it alone cannot scale to larger data sizes.  \sskr provides the evaluation mechanism to achieve this scaling, while using \lcsolve as a building block.  In the experiment of Figure~\ref{fig:inc_tuples}, we evaluate our query workload with all three methods, while increasing the data size to several million tuples.  While \lcsolve continues to outperform \ssearch, their runtimes increase sharply when the data size is higher than $\tau$, and both methods fail to produce results on data larger than 1M tuples (or less in some cases), due to time-outs (runtime exceeds 1.5 hours) or crashes (Gurobi runs out of memory).  In contrast, \sskr maintains remarkably stable performance, thanks to its divide-and-conquer approach, ensuring that ILP subproblems fit within the available memory. On relations with size less than $\tau$, \sskr and \lcsolve have identical performance as \sskr makes a single \lcsolve call with the entire relation.  The runtime benefits do not impact quality, and \sskr produces packages that approximate the continuous relation of each SPQ with ratios ranging from $0.92$ to $0.97$.  For the cases that \ssearch and \lcsolve produced results, those were also of high quality, with approximation ratios in the 0.95--0.98 range.


Query run times are higher for the Portfolio workload than TPC-H due to the expensive nature of sampling from a Geometric Brownian Motion model than from a Guassian one. Even so, since \sskr requires fewer optimization scenarios and does not generate scenarios from all tuples at once, its runtime on the portfolio datasets stays under 20 minutes for all queries.\footnote{While validation scenarios are larger in number, they only need to be generated for few tuples, so they do not significantly impact runtime.}

\smallskip
\noindent
\fbox{
\parbox{0.96\columnwidth}{
\emph{Key takeaway:} 
   \sskr demonstrates remarkable scaling to very large data, significantly outperforming the state of the art.
}}

\section{Related Work} \label{sec:related_works}

\textbf{Probabilistic Databases}~\cite{CA22,dalvi2007efficient, JampaniXWPJH11} specialize in representing data with stochastic attributes. Prior work has focused primarily on supporting SQL-type queries over uncertain data, allowing for ad hoc what-if analyses, as opposed to the in-database optimization that we support, which systematically searches through the space of possible decisions to find the optimal course of action.


\textbf{In-database decision-making} pushes decision making-related activities closer to where the data resides, and thus simplifies decision-making workflows and reduces unnecessary data movement. The work in \cite{CA22} supports what-if analysis over historical data, but does not support full-scale optimization.
SolveDB~\cite{vsikvsnys2016solvedb} and its successor, SolveDB+~\cite{siksnys2021solvedb+} integrate a variety of black-box optimizers and their functionalities into a DBMS. However, they do not scale to large problems; presumably some package query solving techniques could be incorporated into SolveDB+.
\reva{
\skr~\cite{brucato2018package} scales \emph{deterministic} package queries using a partitioning and divide-and-conquer strategy similar to \sskr.  However, the latter introduces critical nuances (e.g., the duplicate representatives) to handle uncertainty and correlations.
} 
Recently, \pshading~\cite{mai2023scaling} was able to improve over \skr using a novel hierarchical partitioning scheme together with customized ILP solvers. An interesting direction for future work would be to extend this approach to stochastic data.


\textbf{Stochastic Optimization} has often used Monte Carlo sampling to estimate uncertain attributes that cannot be quantified exactly~\cite{homem2014monte}. Constrained optimization involving probabilistic constraints can be hard to handle, as they can make the feasible regions non-convex~\cite{ahmed2008solving}. Sample Average Approximation generates many scenarios to quantify these constraints using sample values~\cite{kim2015guide}. 
Prior methods have often attempted to address this difficulty by reducing the number of scenarios~\cite{BrucatoYAHM20, campi2011sampling}. 
\reva{\ssearch~\cite{BrucatoYAHM20} does so via use of conservative summaries as described previously, restricting the feasible region improve feasibility.
However, the feasible regions for the optimization problems remain nonconvex in general, and the search for an optimal set of summaries can be time consuming.} Other methods have approximated non-convex risk constraints using their convex CVaR counterparts, but these approaches incur a quadratic time complexity for doing so~\cite{bubeck2015convex}. Our \textsc{RCL-Solve} approach effectively summarizes all of the information contained in numerous scenarios with one linear constraint in the ILP, and finds the \reva{best set of L-CVaR constraints
within a logarithmic number of iterations using bisection.}

\section{Conclusion and Future Directions} 
\label{sec:conclusion}
In this work, we propose \textsc{RCL-Solve}, an improved SPQ solver that converts non-convex stochastic constrained optimization problems into ILPs whose size is independent of the number of scenarios. With \sskr, we tackle scalability challenges in constrained optimization caused by the number of tuples in large relations. We further introduce \dpart, an efficient approach for partitioning probabilistic relations. Together, these novel methods allow users, in the face of uncertainty and large datasets, to quickly compute near-optimal packages that satisfy a given set of constraints. 
In the future we wish to explore (i) how to relax stochastic query constraints to recommend alternative packages with better objectives that may slightly violate current constraints, (ii) how to scalably support sequential decision-making or two-stage stochastic programs, and (iii) how to explain SPQ solutions allowing for increased user trust with uncertain data.

\begin{acks}
This work was supported by the ASPIRE Award for Research Excellence (AARE-2020) grant
AARE20-307 and NYUAD CITIES, funded by Tamkeen under the Research Institute Award CG001, and by the National Science Foundation under grants 1943971 and 2211918.
\end{acks}

\clearpage

\bibliographystyle{ACM-Reference-Format}
\bibliography{sample}
\clearpage
\appendix

\ifextended
\else
\noindent\textbf{[THE FOLLOWING APPENDICES WILL APPEAR IN THE ONLINE EXTENDED VERSION ONLY]}\\
\fi

\section{Frequently Used Notations}
\label{appendix:notations}
\revc{For the reader's convenience, Table~\ref{tab:notations} displays a list of frequently used notations that we have used throughout the paper.}

\begin{table}[h]
\caption{\revc{Table of Notations}}
\label{tab:notations}
\begin{tabular}{@{}cc@{}}
\toprule

Notation & Definition\\
\midrule
$\text{CVaR}^\bot_\alpha(A)$ & Lower-tail $\alpha$-confidence CVaR \\
$\text{CVaR}^\top_\alpha(A)$ & Upper-tail $\alpha$-confidence CVaR \\
$\LCVaR_\alpha(x,A)$ & $\sum_{i=1}^n \text{CVaR}_\alpha(t_i.A)*x_i$, for tuples $\{t_1, ..., t_n\}$ \\
$\mathcal{O}$ & Set of Optimization Scenarios\\
$\mathcal{V}$ & Set of Validation Scenarios\\
$\widehat{\text{CVaR}}_\alpha(t_i.A,\mathcal{O})$ & SAA-estimate of $\text{CVaR}_{\alpha}(t_i.A)$ based on $\mathcal{O}$\\
$Q$ & Stochastic Package Query\\
$m$ & Initial Number of Optimization Scenarios\\
$\hat{m}$ & Number of Validation Scenarios\\
$\delta$ & Bisection Termination Threshold\\
$\epsilon$ & Approximation Error Bound\\
$T^R$ & A set of representative tuples\\
$\tau$ & Size Threshold\\
$d_A$ & Diameter Threshold for $A$\\
$\Gamma$ & Risk Tolerance\\
$\Delta \Gamma$ & Change in risk tolerance\\
$P_{max}$ & Maximum number of distinct tuples in a package\\
\bottomrule
\end{tabular}
\end{table}

\section{Proofs of Theorems}
\label{appendix:APSproof}

\begin{proof}[Proof of Theorem~\ref{th:constraintCompare}]
The first implication follows from \cite[Prop.~6.9]{McNeilFE15}. For the second implication, suppose that $\text{VaR}_\alpha(x\cdot A)< V$. Then the monotonicity of $\text{VaR}_\alpha$ implies that $\text{VaR}_u(x\cdot A)< V$ for $u\in [0,\alpha]$. Using the definition in \eqref{eq:defCVaRL}, we then have
\[
\text{CVaR}_\alpha(x\cdot A)=\frac{1}{\alpha}\int_0^\alpha\text{VaR}_u(x\cdot A)\,du
< \frac{1}{\alpha}\int_0^\alpha V\,du=V,
\]
which proves the contrapositive of the implication.
\end{proof}%

\begin{proof}[Proof of Theorem~\ref{th:alternating_parameter_search}]
Since $\mathcal{F}_{Q(S)}$ is nonempty, we know that $(\alpha^*,V^*)$ exists and so need to show that APS will find this parameterization. First suppose that $\alpha^*<1$ and let  $\tilde\alpha > \alpha^*$ be the least value of $\alpha'$ for which APS finds a parameterization $(\tilde\alpha,\tilde{V})$ yielding an infeasible package. We claim that $\tilde{V}\ge V^*$. To see this, first note that $(\tilde\alpha, V^*)$ must yield an infeasible package since $\tilde\alpha>\alpha^*$ implies that $(\tilde\alpha, V^*)$ yields a higher objective value than $(\alpha^*,V^*)$ and hence the package must be infeasible by definition of $(\alpha^*,V^*)$. However, given $\alpha'=\tilde\alpha$, APS will ensure that $\tilde{V}$ is set to the \emph{maximum} value of $V'$ that yields an infeasible package, due to the pattern of alternating decreases in $\alpha'$ and $V'$ during the search. This proves the claim. After finding $(\tilde\alpha, \tilde{V})$, APS next decreases $\alpha'$ to the maximum value that yields a feasible package---we claim this value must be $\alpha^*$. Indeed, $(\alpha^*, \tilde{V})$ yields a feasible package since $(\alpha^*, V^*)$ does and $\tilde{V} \ge V^*$ by the previous claim, and if $(\hat\alpha,\tilde{V})$ were to be found for some $\tilde\alpha\ge \hat\alpha> \alpha^*$ then APS would decrease $V'$ until a parameterization $(\hat\alpha,\hat{V})$ yielding an infeasible package is found, contradicting the definition of $\tilde\alpha$ as the \emph{smallest} value of $\alpha'$ exceeding $\alpha^*$ for which PS finds an infeasible solution. Next, APS improves the objective score of the feasible package produced for $(\alpha^*,\tilde{V})$ by reducing $V'$ as much as possible. Since $V'=V^*$ results in the most optimal feasible package, APS will find the parametrization $(\alpha^*, V^*)$.

Now suppose that $\alpha^*=1$. Since $V^* \leq V$ implies that $(1,V)$ is at least as restrictive as $(1,V^*)$, the package corresponding to $(1,V)$ is feasible. Thus the first step of APS, i.e., the $\alpha$-search, will immediately terminate and return the pair $(1,V)$. The $V$-search will then find the pair $(1,V^*)=(\alpha^*,V^*)$ and terminate. Thus APS will find $(\alpha^*,V^*)$ in all cases.
\end{proof}

\reva{\begin{proof}[Proof of Theorem~\ref{th:mad_cvar}]
For $i=1,2$ denote by $E_i$ the event $\{t_i.C \le q_{\alpha}(t_i.C)\}$. Observe that 
\begin{align*}
\Pr(E_1\setminus E_2)&=\Pr(E_1)-\Pr(E_1 \cap E_2)\\
&=\alpha-\Pr(E_1 \cap E_2)\\
&=\Pr(E_2)-\Pr(E_1 \cap E_2)\\
&=\Pr(E_2 \setminus E_1).
\end{align*}
Hence, we define $c=\Pr(E_1\setminus E_2)=\Pr(E_2 \setminus E_1) > 0$. Next observe that
\begin{align*}
    &\mathbb{E}[t_2.C|E_1] - \mathbb{E}[t_2.C|E_2]\\ 
    &\qquad=\frac{\mathbb{E}[t_2.C\times\mathbf{1}(E_1)]}{\Pr[E_1]}-\frac{\mathbb{E}[t_2.C\times\mathbf{1}(E_2)]}{\Pr[E_2]}\\
    &\qquad=\frac{1}{\alpha}(\mathbb{E}[t_2.C\times(\mathbf{1}(E_1)-\mathbf{1}(E_2))])\\
    &\qquad=\frac{1}{\alpha}(\mathbb{E}[t_2.C\times(\mathbf{1}(E_1\setminus E_2)-\mathbf{1}(E_2 \setminus E_1))])\\
    &\qquad=\frac{1}{\alpha}(\mathbb{E}[t_2.C\times\mathbf{1}(E_1\setminus E_2)-t_2.C \times \mathbf{1}(E_2 \setminus E_1)])\\
    &\qquad=\frac{c}{\alpha}\left(\frac{\mathbb{E}[t_2.C\times\mathbf{1}(E_1\setminus E_2)]}{\Pr(E_1\setminus E_2)}-\frac{\mathbb{E}[t_2.C \times \mathbf{1}(E_2 \setminus E_1)]}{\Pr(E_2 \setminus E_1)}\right)\\
    &\qquad=\frac{c}{\alpha}\left(\mathbb{E}[t_2.C|E_1\setminus E_2]-\mathbb{E}[t_2.C|E_2\setminus E_1]\right) \geq 0,
\end{align*}
where the last inequality holds because 
the first expectation integrates over $t_2.C$ values that exceed $q_{\alpha}(t_2.C)$ whereas the second expectation integrates over values equal to at most $q_{\alpha}(t_2.C)$. Thus $\mathbb{E}[t_2.C|E_1] \geq \mathbb{E}[t_2.C|E_2]=\text{CVaR}_{\alpha}(t_2.C)$ and it follows that
Then we have 
\begin{align*}
  \text{CVaR}_{\alpha}(t_1.C) & = \mathbb{E}[t_1.C \bigm\vert E_1]\\
  &= \mathbb{E}[t_2.C \bigm\vert E_1] - \mathbb{E}[t_2.C - t_1.C \bigm\vert E_1]\\
  &\ge \mathbb{E}[t_2.C \bigm\vert E_1]- \mathbb{E}[|t_2.C - t_1.C| \bigm\vert E_1]\\
  &\ge \text{CVaR}_{\alpha}(t_2.C) - \mathbb{E}[|t_2.C - t_1.C| \bigm\vert E_1]\\
  &\ge \text{CVaR}_{\alpha}(t_2.C) - \frac{\mathbb{E}[|t_2.C - t_1.C|]}{\alpha}\\
  &\ge \text{CVaR}_{\alpha}(t_2.C) - \frac{d_C}{\alpha}.
\end{align*}

Hence
\[
\text{CVaR}_{\alpha}(t_2.C) - \text{CVaR}_{\alpha}(t_1.C) \le \frac{d_C}{\alpha}.
\]

Similarly, we can show that
\[
\text{CVaR}_{\alpha}(t_1.C) - \text{CVaR}_{\alpha}(t_2.C) \le \frac{d_C}{\alpha},
\]
so that
\[
|\text{CVaR}_{\alpha}(t_1.C) - \text{CVaR}_{\alpha}(t_2.C)| \le \frac{d_C}{\alpha}.
\]

\end{proof}}

\begin{proof}[Proof of Theorem~\ref{th:opt_sskr}]
Write $S=\{t_1,\ldots,t_n\}$ and $S_0=\{r_1,\ldots,r_d\}$. Consider first the sketch phase and observe that, since duplicates of a representative have identical marginal distributions, we have that, for any $\lambda\in\Lambda$, 
\[
\begin{split}
    \mathbb{E}[t_i.O] - \mathbb{E}[\lambda(t_i).O] &= \mathbb{E}[t_i.O - \lambda(t_i).O]\\
    &\le \mathbb{E}[|t_i.O - \lambda(t_i).O|] \le d_o.
\end{split}
\]
Hence $\mathbb{E}[\lambda(t_i).O] \ge \mathbb{E}[t_i.O] - d_o$. By assumption, \sskr returns a non-NULL package $x^*$, and the corresponding objective value for $y=\lambda(x^*)$ in $Q(S_0)$ is
\[
\begin{split}
    \omega^*_s &= \sum_{j=1}^d \mathbb{E}[r_j.O]y_j \\
    &= \sum_{i=1}^{n} \mathbb{E}[\lambda(t_i).O]x^*_i \ge \sum_{i=1}^{n} (\mathbb{E}[t_i.O]-d_o)x^*_i\\
    &= \sum_{i=1}^{n} \mathbb{E}[t_i.O]x^*_i - d_o \sum_{i=1}^{n} x^*_i \ge \omega^* - d_o P_{max}.
\end{split}
\]
Since $x^*$ is assumed to be sketch feasible, we can choose $\lambda$ so that $y=\lambda(x^*)$ is a feasible solution of $Q(S_0)$. Then the objective value $\omega^*_{sk}$ for the optimal solution $y^*$ to $Q(S_0)$ satisfies
$\omega^*_{sk} \ge \omega^*_s\ge\omega^* - d_o P_{max}$. Note that $\omega^*_{sk}  \ge \omega^*_{s}$ may hold even if $x^*$ is not sketch feasible, in which case the approximation guarantee will still hold. As discussed before the statement of the theorem, the objective value $\Bar{\omega}$ for the solution to $Q(S_0)$ returned by \textsc{SolveSketch} satisfies $\Bar{\omega} \ge (1-\epsilon)\omega^*_{sk}\ge (1-\epsilon)(\omega^* - d_o P_{max})$. Moreover, the objective value $\omega$ corresponding to the solution $x^*$ satisfies $\omega \ge (1-\epsilon) \Bar{\omega}$, so that
%
\begin{align*}
\omega \ge (1-\epsilon) \Bar{\omega} \ge (1-\epsilon)^2(\omega^* - d_o P_{max})    
\end{align*}
as desired.
\end{proof}

\section{Accelerated NORTA Procedure}
\label{appendix:NORTA}

\begin{algorithm}[t]
\caption{\textsc{GenerateScenario}}\label{alg:scenario_construction}
\begin{algorithmic}[1]
\Require $t :=$ A representative tuple
\Statex $A :=$ A stochastic attribute
\Statex $F_{(t)} :=$ The cumulative distribution function of $t.A$
\Statex $d :=$ The number of duplicates
\Statex $\overline{\rho} :=$ Correlation coefficient between duplicates
\State $F_{(t)} \gets$ Estimated CDF of $t.A$
\State $\overline{\kappa}=\textsc{NORTAfit}(\overline{\rho},F_{(t)})$ \Comment{compute $\Sigma_Z$ off-diagonal entry value}\label{li:NORTAfit}
\State $s_1, \dots, s_d \gets$ i.i.d. samples from $N(0, 1)$
\State $s'_1 \gets s_1 \sqrt{1 + (d-1) \overline{\kappa}}  - \sum_{i=2}^{d} s_i \sqrt{1 - \overline{\kappa}}$
\For {$i \in \{2, \dots, m\}$}
    \State $s'_i \gets s_i \sqrt{1 - \overline{\kappa}} - s_1 \sqrt{1 + (d-1) \overline{\kappa}}$
\EndFor
\For {$i \in \{1, \dots, d\}$}
    $t^{(i)}.A = F_{(t)}^{-1}(\Phi(s'_i))$
\EndFor
\State \Return $[t^{(1)}.A,\ldots,t^{(d)}.A]$

\end{algorithmic}
\end{algorithm}

Let $t$ be a representative, $A$ a stochastic attribute of interest, and $\{t^{(1)},\ldots,t^{(d)}\}$ a set of $d$ duplicates of $t$. To generate a scenario for the random variables  $t^{(1)}.A,\ldots,t^{(d)}.A$, which are mutually correlated with common pairwise correlation coefficient $\overline{\rho}$, we propose an accelerated version of the NORTA (NORmal-To-Anything) copula method~\cite{ghosh2003behavior}. In general, to generate samples of a $d$-dimensional random vector $Y=(Y_1,\ldots,Y_d)$ having specified marginal cumulative distribution functions (CDF's) $F_1,\ldots,F_d$ and a specified covariance matrix $\Sigma_Y$, the NORTA method first generates a $d$-dimensional multivariate normal random variable $Z$ with mean vector $(0,\ldots,0)$ and covariance matrix $\Sigma_Z$, and then generates a sample $Y$ by setting $Y_i=F^{-1}_i\bigl(\Phi(Z_i)\bigr)$ for $i=1,\ldots,d$, where $F^{-1}_i(u)=\inf\{y: F_i(y)\ge u\}$; standard results show that each $Y_i$ has CDF $F_i$. The covariance matrix $\Sigma_Z$ is selected (via offline solution of a semi-definite program) so that the resulting covariance matrix of $Y$ is either very close to, or exactly equal to, the target covariance matrix $\Sigma_Y$. To generate the $Z$ vector, we first generate a $d$-dimensional vector $Z'$ of mutually independent $N(0,1)$ random variables, and then set $Z=U\sqrt{\Lambda} Z'$, where $\Lambda$ is a diagonal matrix whose entries are the eigenvalues of $\Sigma_Z$ and $U$ is an orthogonal matrix whose columns are eigenvectors of $\Sigma_Z$; see \cite{gentle2009computational}. In general, computing $\Lambda$ and $U$ has a computational complexity of $\mathcal{O}(d^3)$. This makes constructing scenarios expensive as we need to perform this decomposition for every representative. Furthermore, the algorithm may iteratively increase the number of duplicates to get better packages, requiring the computation to be repeated at each iteration. Fortunately, $\Sigma_Y$ has a special structure in our case, namely, $\Sigma_Y(i,j)$ equals 1 if $i=j$ and equals $\overline{\rho}$ otherwise. We can exploit this structure to significantly speed up the NORTA calculations.




In more detail, for any pair $(i,j)$ with $i\not=j$, the general NORTA method computes $\Sigma_Z(i,j)$ as a solution $\kappa_{i,j}$ to the equation
\begin{equation}\label{eq:rootFind}
\begin{split}
\Sigma_Y(i,j)&=
\int_{-\infty}^{\infty} \int_{-\infty}^{\infty} F_i^{-1}\left(\Phi\left(z_i\right)\right) F_j^{-1}\left(\Phi\left(z_j\right)\right)\\
 &\qquad\cdot \varphi_{i j}\left(z_i, z_j,\kappa_{i,j}\right) \,dz_i \,dz_j-\mu_i \mu_j,
\end{split}
\end{equation}
where $\mu_i$ and $\mu_j$ are the means of $F_i$ and $F_j$ and
\[
\begin{split}
&\varphi_{i j}(z_i,z_j,\kappa_{i,j})\\
&\quad=\frac{1}{2 \pi \sqrt{1-\kappa_{i,j}^2}} \exp \left\{-\frac{1}{2\left(1-\kappa_{i,j}^2\right)}\left[z_i^2-2 \kappa_{i,j} z_i z_j+z_j^2\right]\right\}
\end{split}
\]
is the standard bivariate normal probability density function. In our setting, $F_i\equiv F_{(t)}$ for all duplicates, where $F_{(t)}$ is the CDF of the corresponding random variable $t.A$, and $\Sigma_Y(i,j)\equiv\bar\rho$ for all pairs $(i,j)$ with $i\not= j$. Thus the root-finding problem in \eqref{eq:rootFind} is identical for all $(i,j)$ pairs, so that we can take $\kappa_{i,j}\equiv\bar\kappa$ for an appropriate value of $\bar\kappa\ge 0$. This gives us the matrix $\Sigma_Z$ after setting $\Sigma_Z(i,i)=1$ for $i\in[1..d]$. A closed form of the eigenvalues and eigenvectors of $\Sigma_Z$ can then be trivially computed as $\lambda_1 = 1+(d-1)\overline{\kappa}$ with the corresponding eigenvector $[1, 1, \dots, 1]$ and $\lambda_2 = 1-\overline{\kappa}$ with $d-1$ corresponding eigenvectors $[-1, 1, 0, \dots, 0], [-1, 0, 1, \dots, 0],\dots,[-1, 0, \dots, 0, 1]$.

Thus, after generating a vector $s=[s_1, \dots, s_d]$ of i.i.d.\ samples from the standard normal distributions, a sample vector $s'$ of $d$ correlated multivariate normal distributions can be computed in $\mathcal{O}(d)$ time as $s' = [s_1\sqrt{\lambda_1} - \sqrt{\lambda_2}\sum_{i=2}^{d}s_i, s_2\sqrt{\lambda_2} - s_1\sqrt{\lambda_1}, \dots, s_d\sqrt{\lambda_2} - s_1\sqrt{\lambda_1}]$. Finally, sample values $t^{(1)}.A,\ldots,t^{(d)}.A$ can be generated as $t^{(i)}.A=F_{(t)}^{-1}\bigl(\Phi(s'_i)\bigr)$ for $i\in[1..d]$.
The CDF $F_{(t)}$ can be pre-computed offline using a histogram-based density estimation scheme as in~\cite{chan2014near}. Algorithm~\ref{alg:scenario_construction} describes our linear time NORTA-based scenario construction scheme. The function \textsc{NORTAfit} in line~\ref{li:NORTAfit} solves the root-finding problem in \eqref{eq:rootFind}.

\section{Partitioning Details}
\label{sec:DPdetails}

\subsection{MAD Case \revision{Studies}} \label{appendix:mad_case_study}

To gain more insight into the behavior of the MAD distance measure, we consider two different scenarios.

\smallskip
\noindent{\textbf{Correlated normal random variables.}}
We first consider the case of two correlated normal random variables $X$ and $Y$ having respective means $\mu_X$ and $\mu_Y$, respective variances $\sigma^2_X$ and $\sigma^2_Y$, and covariance $c_{XY}$. Let $\Delta=\mathbb{E}[X-Y]=\mu_x-\mu_y$ and $\sigma^2=\text{Var}[X-Y]=\sigma^2_X-2c_{XY}+\sigma^2_T$. Using standard results for the ``folded normal'' distribution, we have that
\begin{equation}\label{eq:MADnorm}
\textsc{MAD}(X,Y)=\sigma\sqrt{2/\pi}e^{-\Delta^2/2\sigma^2}+\Delta \bigl(1-2\Phi(-\Delta/\sigma)\bigr),
\end{equation}
where $\Phi$ is the standard normal CDF. If  $\sigma^2$ is very small relative to $\Delta$, then $\textsc{MAD}\approx \Delta$, and thus the MAD distance primarily reflects the distance between the distribution means. If $\sigma^2$ is much larger than $\Delta$, then $\textsc{MAD}\approx \sigma\sqrt{2/\pi}$ and thus mostly depends on the variance of $X-Y$. In this latter case, $\sigma^2$ becomes smaller as nonnegative correlation between $X$ and $Y$ increases, so that the MAD distance primarily reflects the correlation between $X$ and $Y$. The general formula in \eqref{eq:MADnorm} balances inter-mean distance and correlation.

\begin{figure}
    \centering
    \includegraphics[width=1\linewidth]{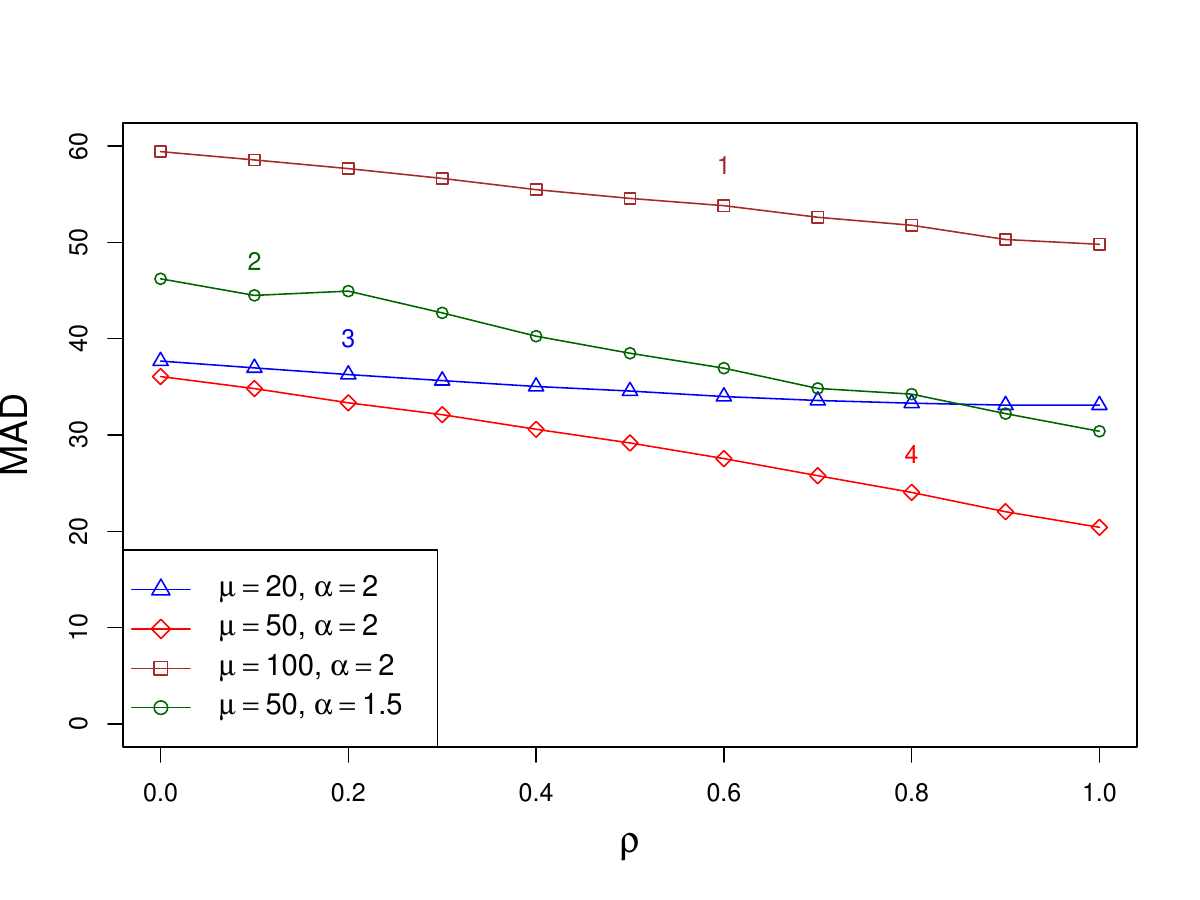}
    \caption{\reva{MAD for correlated Pareto$(a,\alpha)$ and uniform$[0,100]$ random variables for different values of NORTA correlation coefficient $\rho$, Pareto mean $\mu=\bigl(\alpha/(\alpha-1)\bigr)a$, and Pareto tail coefficient $\alpha$.}}
    \label{fig:ParetoUniform}
\end{figure}

\smallskip
\noindent\reva{\textbf{Correlated Pareto and uniform random variables.}
We next consider a scenario where the two correlated random variables $(X,Y)$ have widely different marginal distributions, one of which also has a heavy tail, and show that the MAD distance behaves in the desired manner.

Let the marginal distribution of $X$ be a Pareto distribution with CDF $F_X(x)=\bigl(1-(x/a)^{-\alpha}\bigr)I[x\ge a]$, where $a\in(0,\infty)$, $\alpha\in(1,2]$, and $I[\cdot]$ is an indicator function. For $\alpha$ in the given range, $\mathbb{E}[X]=\bigl(\alpha/(\alpha-1)\bigr)a<\infty$ but $\text{Var}[X]=\infty$. Let the marginal distribution of $Y$ be a uniform distribution on $[0,100]$, i.e., $F_Y(x)=\min(x/100,1)I[x\ge 0]$, so that $\mathbb{E}[Y]=50$ and $Y$ has finite moments of all orders. Suppose that $X$ and $Y$ are correlated using the NORTA method; that is, to generate a joint sample of $X$ and $Y$, first generate $(Z_1,Z_2)$ according to a bivariate normal distribution with covariance matrix $\Sigma=\big(\begin{smallmatrix} 1 & \rho\\ \rho & 1 \end{smallmatrix}\big)$, then set $(U,V)=\bigl(\Phi(Z_1),\Phi(Z_2)\bigr)$ where $\Phi$ is the CDF of a $N(0,1)$ random variable, and finally set $(X,Y)=\bigl(F_X^{-1}(U),F_Y^{-1}(V)\bigr)$. Thus the degree of correlation between $X$ and $Y$ is governed by the correlation $\rho$ of the bivariate normal distribution: if $\rho=0$, then $X$ and $Y$ are mutually independent and if $\rho=1$, then $X$ and $Y$ are perfectly positively correlated in that $Y=g(X)$ for some increasing deterministic function $g$.

Figure~\ref{fig:ParetoUniform} shows the MAD between $X$ and $Y$ for various values of $\rho$, $\alpha$, and $\mu=\mathbb{E}[X]$. As can be seen, for all values of $\mu$ and $\alpha$, the MAD decreases as $\rho$ increases, that is, as the correlation between $X$ and $Y$ increases. Comparing curves~1, 3, and 4, we see that, for a given value of $\rho\in[0,1]$ and given value $\alpha=2$, the MAD decreases as $|\mathbb{E}[X]-\mathbb{E}[Y]|$ decreases. Finally, comparing curves 2 and 4, we see that, for fixed $\rho\in[0,1]$ and $\mathbb{E}[X]=\mathbb{E}[Y]=50$, the MAD increases as the tail coefficient,  and hence the variability, of the Pareto distribution increases relative to the variability of the marginal distribution of $Y$.

\begin{figure}
    \centering
    \includegraphics[width=0.8\linewidth]{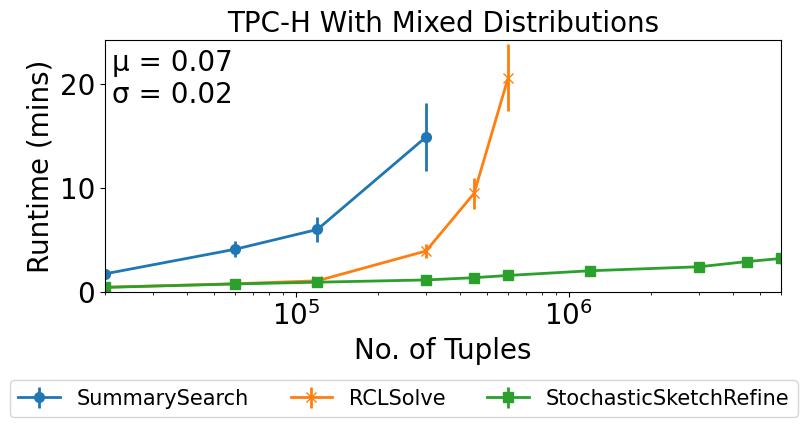}
    \caption{\reva{Consistent with the notation in Figure~\ref{fig:inc_tuples}, ($\mu, \sigma$) shows the relative integrality gap statistics for \sskr's packages. In a setup where underlying distributions for different tuples are heterogeneous for the same attribute, MAD-created partitions continue to have tuples similar enough to create packages of good quality, and balance the number of partitions with the number of tuples-per-partition to allow \sskr to outperform \lcsolve and \ssearch. Absence of data points indicates no solutions were generated in 1.5 hours.}}
    \label{fig:mixed_dis_runtime}
\end{figure}

\noindent\textbf{Performance of \sskr with a heterogeneous stochastic attribute.}
To further examine how effectively MAD identifies similar tuples within heterogeneous and heavy-tailed distributions in the context of \sskr, we repeated the experiment in Figure~\ref{fig:inc_tuples} for a single query with a modified version of the stochasticized Lineitem relation of the TPC-H dataset. As before, samples of the stochastic attribute `Price' are obtained for a given tuple by adding to the original deterministic TPC-H Price value a noise term sampled from a tuple-specific probability distribution. For half of the tuples, this distribution was a Pareto distribution with a tuple-specific mean and $\alpha$ sampled from $U(1, 5)$ and $U(1, 2.5)$; for the remaining tuples, the distribution was a uniform distribution $U(l, h)$, where tuple-specific values of $h$ and $l$ were randomly sampled from $[-10, 10]$ and $[-10, h]$. The stochastic attribute `Quantity' was specified in the same way.
We kept all hyperparameters settings the same as those described in Section~\ref{sec:experimental_evaluation}. For a sample query, Figure~\ref{fig:mixed_dis_runtime} shows that the partitions generated with MAD were both small and similar enough to create packages from large datasets within a reasonable amount of time and with small integrality gaps despite the heterogeneity in the underlying distributions of the tuples and the heavy-tailed nature of the Pareto distribution.
}

\subsection{MAD Estimation} \label{appendix:mad_estimation}

\begin{figure}
    \centering
    \includegraphics[width=0.8\linewidth]{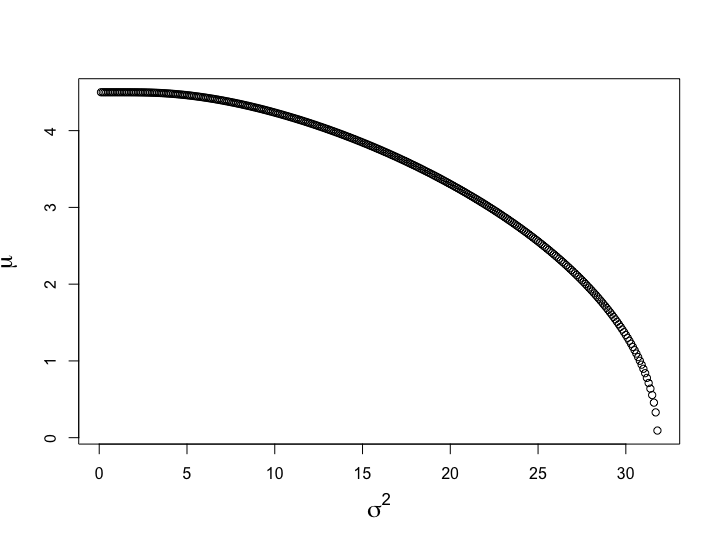}
    \caption{Trade-off between values of $\mu$ and $\sigma^2$ when maintaining a fixed value of $\mu_D=d+\epsilon$ ($d=4$ and $\epsilon=0.5$).}
    \label{fig:muSig2Tradeoff}
\end{figure}

\begin{figure}
    \centering
    \includegraphics[width=0.75\linewidth]{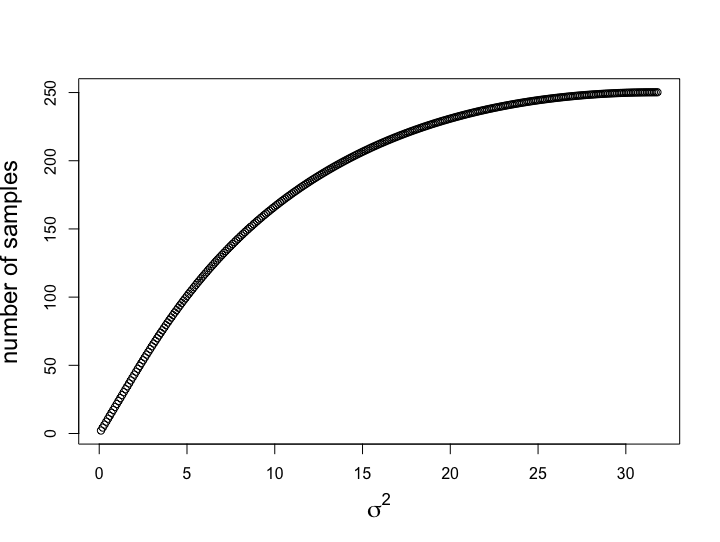}
    \caption{Minimum required sample size as a function of $\sigma^2$ ($d=4$, $\epsilon=0.5$, $p=0.01$). For each value of $\sigma^2$, we compute a corresponding value of $\mu$ such that $\mu_D(\mu,\sigma^2)=d+\epsilon$.}
    \label{fig:nsampCurve}
\end{figure}

To estimate $\textsc{MAD}(X,Y)$ for two (possibly dependent) attributes $X$ and $Y$ using a sample size of $n>1$, we take i.i.d.\ samples $(X_1,Y_1),\ldots,(X_n,Y_n)$, compute the quantities $D_i=|X_i-Y_i|$ for $i\in[1..n]$, and then estimate the true value $m$ of $\textsc{MAD}(X,Y)$ as $\hat m_n=n^{-1}\sum_{i=1}^n D_i$. To help figure out roughly how large $n$ needs to be, we reason as follows. When estimating $m$ in the context of \dpart, the most important goal is to avoid erroneously including distant points in a given partition. We can formalize this goal as ensuring that $\mathbb{P}(\hat m_n\le d)\le p$ when $m=d+\epsilon$, where $d$ is the target diameter of the partition, $p$ is a specified small probability, and $\epsilon>0$ represents an ``indifference zone'' within which we are willing to accept a misclassification error. For tractability, suppose that $X$ and $Y$ have normal distributions with means $\mu_X$ and $\mu_Y$, variances $\sigma^2_X$ and $\sigma^2_Y$, and covariance $\sigma_{XY}$. Without loss of generality, suppose that $\mu_X\ge\mu_Y$. Then $D_i$ has mean
\[
\mu_D=\sqrt{2\sigma^2/\pi}e^{-\mu^2/2\sigma^2}+\mu \bigl(1-2\Phi(-\mu/\sigma)\bigr)
\]
and variance $\sigma_D^2=\mu^2+\sigma^2-\mu_D^2$, where $\mu=\mu_X-\mu_Y$ and $\sigma^2=\sigma_X^2+\sigma_Y^2-2\sigma_{XY}$. The central limit theorem now implies that $M_n\stackrel{D}{\sim}N(\mu_D,\sigma_D^2/n)$ to a good approximation as $n$ becomes large. Thus
\[
\mathbb{P}(M_n\le d)\approx\Phi\left(\frac{d-\mu_D}{\sigma_D/\sqrt{n}}\right),
\]
where $\Phi$ is the CDF of a standard $N(0,1)$ normal random variable. We are assuming that $\mu$ and $\sigma^2$ are such that $m=\mu_D=d+\epsilon$, so that
\[
n\ge \left(\frac{\sigma_D z_{1-p}}{\epsilon}\right)^2,
\]
where
\[
\sigma_D^2=\mu^2+\sigma^2-(d+\epsilon)^2
\]
and $z_q=\Phi^{-1}(q)$. When maintaining a fixed value of $\mu_D=\mu_D(\mu,\sigma^2)$, parameters $\mu$ and $\sigma^2$ vary inversely; see Figure~\ref{fig:muSig2Tradeoff}. Possible values of $(\mu,\sigma^2)$ range from $(0,\bar\sigma^2)$ to $(d+\epsilon,0)$, where $\bar\sigma=(d+\epsilon)\sqrt{\pi/2}$. To compute the most conservative requirement on the number of samples for given values of $d$, $\epsilon$, and $p$, we want to choose $\mu$ and $\sigma^2$ to maximize $\sigma_D^2=\sigma_D^2(\mu,\sigma^2)$ subject to the constraint $\mu_D(\mu,\sigma^2)=d+\epsilon$. As can be seen from Figure~\ref{fig:nsampCurve}, the maximum value of $n$ occurs when $\sigma^2=\bar\sigma^2$ and $\mu=0$. In this case the number of required samples is
\[
n\approx\frac{(0.5\pi-1)(d+\epsilon)^2z_{1-p}^2}{\epsilon^2}.
\]
For example, during our experiments, we set the diameter threshold for the attribute gain in the portfolio dataset to $d=100$. Considering $\epsilon=10$, and $p=0.05$, we have $n\approx 188$. 

\subsection{\dpart Details} \label{appendix:dpart_details}
The pseudocode for the \textsc{\dpart} routine is given as Algorithm~\ref{alg:dpart}.

\begin{algorithm}[p]
\caption{\textsc{\dpart}}\label{alg:dpart}
\begin{algorithmic}[1]
\Require $R :=$ A set of stochastic tuples
\Statex $\tau :=$ Size Threshold
\Statex $\hat{d}=[d_{A_1}, d_{A_2}, \dots, d_{A_d}] :=$ Diameter Threshold for each attribute

\Ensure $P :=$ Partitioning of $R$ satisfying size and diameter constraints  

\State $pivot :=$ \textsf{PickAtRandom($R$)}
\State $r_{max}:= 0$ \Comment{Highest diameter-to-threshold ratio}
\State $A_{max}:= $ None \Comment{Attribute with highest ratio}
\For{$i=\{1, \dots, d\}$}
\State \textsf{\_}, $d^{R}_{A_i} \gets$ \textsf{PivotScan($pivot, A_i$)} \Comment{Distance of the farthest tuple from $pivot$ with respect to $A_i$}
\If{$\frac{2*d^R_{A_i}}{d_{A_i}} \ge r_{max}$}
\State $r_{max} \gets \frac{2*d^{R}_{A_i}}{d_{A_i}}$
\State $A_{max} \gets A_i$ \Comment{Update attribute with highest ratio}
\EndIf
\EndFor
\If{$|R| \le \tau$ and $r_{max} \le 1$}
\State \Return $R$ \Comment{Both size and diameter constraints are satisfied}
\Else
\State \textsf{ids\_with\_distances}, \textsf{\_} $\gets$ \textsf{PivotScan($pivot, A_{max}$)} \Comment{Get tuples identifiers with increasing distances from $pivot$}
\State $pivot \gets$ \textsf{LastElement}(\textsf{ids\_with\_distances}) \Comment{Repivot to the farthest tuple}
\State \textsf{ids\_with\_distances}, \textsf{\_} $\gets$ \textsf{PivotScan($pivot, A_{max}$)} \Comment{Get tuples identifiers with increasing distances from $pivot$}
\State $P \gets \phi$
\State \textsf{CSet} $\gets \phi$
\If{$|R|>\tau$} \Comment{Size-based Partitioning}
\For{$id, \_ \in$ \textsf{ids\_with\_distances}}
\If{|\textsf{CSet}|$=\tau$} \Comment{\textsf{CSet} is already full}
\State $P \gets P$ $\cup$ \dpart(\textsf{CSet}, $\tau$, $\hat{d}$) \Comment{Recursively repartition current set of tuples}
\State \textsf{CSet} $\gets \phi$
\EndIf
\State \textsf{CSet} $\gets$ \textsf{CSet} $\cup$ $id$
\EndFor
\If{|\textsf{CSet}| $> 0$}
\State $P \gets P$ $\cup$ \dpart(\textsf{CSet}, $\tau$, $\hat{d}$) \Comment{Repartition remaining tuples}
\EndIf
\Else \Comment{Distance-based Partitioning}
\State \textsf{multiplier} $\gets 1$
\For{\textsf{id}, \textsf{distance} $\in$ \textsf{ids\_with\_distances}}
\If{\textsf{distance} $\ge$ \textsf{multiplier}$*d_{A_{max}}$}
\State $P \gets P$ $\cup$ \dpart(\textsf{CSet}, $\tau$, $\hat{d}$) \Comment{Recursively repartition current set of tuples}
\State \textsf{CSet} $\gets \phi$
\State \textsf{multiplier} $\gets$ \textsf{multiplier} $+1$
\EndIf
\State \textsf{CSet} $\gets$ \textsf{CSet} $\cup$ $id$
\EndFor
\If{|\textsf{CSet}| $> 0$}
\State $P \gets P$ $\cup$ \dpart(\textsf{CSet}, $\tau$, $\hat{d}$) \Comment{Repartition remaining tuples}
\EndIf
\EndIf
\State \Return $P$
\EndIf
\end{algorithmic}
\end{algorithm}

\subsection{Representative Selection Details} \label{appendix:rep_details}
The pseudocode for the \textsc{RepresentativeSelection} routine is given as Algorithm~\ref{alg:select}.

\begin{algorithm}[t]
\caption{\textsc{RepresentativeSelection}}\label{alg:select}
\begin{algorithmic}[1]
\Require $A :=$ the stochastic attribute
\Statex $t_1,\dots,t_N :=$ N tuples
\Statex $S :=$ the set of precomputed scenarios 

\For {$j \in \{1,\dots,n\}$ \textbf{and} $i \in \{1,\dots,N$\}}
    \State $S_{ij}.A \gets $ the realized value of $t_i.A$ in scenario $j$\label{ln:tAval}
\EndFor

\For {$j \in \{1,\dots,n\}$} \label{ln:forloopMinMax}
    \State $m_j.A, M_j.A \gets \min_{1\le i \le N} S_{ij}, \max_{1\le i \le N} S_{ij}$\label{ln:MaxA}
\EndFor

\For {$i \in \{1,\dots,n\}$}
    \State $C_{i}.A \gets \sum_{j=1}^{n} \max(M_j.A - S_{ij}.A, S_{ij}.A- m_j.A)$ \label{ln:tupleCost}
\EndFor

\State $i^* \gets \operatorname*{argmin}_{i} C_i.A$ \label{ln:rep}

\State \Return $t_{i^*}$

\end{algorithmic}
\end{algorithm}

\subsection{Creating Representatives for Correlated Attributes} \label{appendix:correlated_attributes}

The random variables for a set of correlated attributes $\mathbf{A}$ of a representative are made equal to the corresponding variables for the same attributes of any one of the tuples in the partition. We compute the worst-case replacement cost, as described in Algorithm~\ref{alg:select}, for every tuple in the partition for every attribute $A \in \mathbf{A}$, and map the representative's variables for every attribute in $\mathbf{A}$ to those of $t_r$, the tuple in the partition with the least total cost.

While generating scenarios for $d$ duplicates of a representative with $|\mathbf{A}|$ correlated attributes, we aim to generate scenarios so that both inter-tuple and inter-attribute correlations are satisfied. Essentially, we need to generate a total of $d*|\mathbf{A}|$ values for each scenario --- one for each random variable that corresponds to a correlated attribute of a duplicate. The correlation between any two random variables corresponding to different duplicates of the same attribute $A_i \in \mathbf{A}$ should be equal to the median correlation of every tuple in the partition with $t_r$ w.r.t. $A_i$. The correlation between any two random variables that correspond to the same duplicate but different attributes $A_i, A_j \in \mathbf{A}$ is specified by the correlation between $A_i$ and $A_j$. Given these requisite correlations and the histogram estimations of the CDF of each random variable, we can use the NORTA copula method described in Section~\ref{appendix:NORTA} to generate scenarios from the representative.    

\section{Approximation Guarantees} \label{sec:theory}


For any SPQ $Q(S)$, the size and diameter constraints used by \dpart ensure that the objective value $\Bar{\omega}$ of the returned sketch solution is ``close'' to the true objective value $\omega^*$ for $Q(S)$. \lcsolve-R then ensures that the returned objective value at each step of the refine phase is ``close'' to $\Bar{\omega}$. We now formalize this intuition. An execution of \sskr over $Q(S)$ is \emph{successful} if (i)~\textsc{SolveSketch} returns a non-NULL result and (ii)~a refine order is found such that each call to \lcsolve-R returns a non-NULL result; otherwise, \sskr fails and returns NULL. (We did not observe any \sskr failures in our experiments.) In the case of success, we show that for maximization problems, the sketch package produced by \textsc{SolveSketch} will have objective value $\Bar{\omega}\geq (1-\epsilon)\omega^*_{sk}$, where $\omega^*_{sk}$ is the objective value of the optimal sketch solution to $Q(S_0)$. Moreover, as discussed in Section~\ref{subsec:refine}, \textsc{Refine} produces a feasible package with objective value $\omega\geq (1-\epsilon)\Bar{\omega}$. Theorem~\ref{th:opt_sskr} combines these observations to provide an approximation guarantee on the objective value of a non-NULL package $x^*$ produced by \sskr.


We first define ``sketch feasibility''. For a sketch $S_0=\{r_1,\ldots,r_d\}$ of a set $S=\{t_1,\ldots,t_n\}$, denote by $\Lambda$ the set of all functions $\lambda:S\mapsto S_0$ such that, for each $i\in[1..n]$, the element $r_j=\lambda(t_i)$ is one of the duplicate representatives in $t_i$'s partition. For a package $x\in\mathbb{Z}_0^n$ and mapping $\lambda\in\Lambda$ define the package $\lambda(x)=(y_1,\ldots,y_d)\in \mathbb{Z}_0^d$ by setting $y_j=\sum_{i\in I_j}x_i$ for $j\in[1..d]$, where $I_j=\{i: \lambda(t_i)=r_j\}$. A package $x$ for $Q(S)$ is said to be \emph{sketch feasible} if there exists $\lambda\in\Lambda$ such that $y=\lambda(x)$ is a feasible package for $Q(S_0)$.


\begin{theorem}[Approximation Guarantee]\label{th:opt_sskr}
Suppose that \sskr returns a non-NULL solution package $x^*$ for a query $Q(S)$. If $x^*$ is sketch feasible, then 
the objective value $\omega$ corresponding to $x^*$ satisfies
\begin{equation}\label{ex:opt_lb}
\omega\ge (1-\epsilon)^2(\omega^* - d_o P_\text{max}),
\end{equation}
where $\omega^*$ is the objective value  of the optimal package solution to $Q(S)$, the quantity $P_\text{max}$ is the maximum allowable package size, and $d_o$ is the diameter threshold of the objective attribute $O$.
\end{theorem}

\begin{proof}
Write $S=\{t_1,\ldots,t_n\}$ and $S_0=\{r_1,\ldots,r_d\}$. Consider first the sketch phase and observe that, since duplicates of a representative have identical marginal distributions, we have that, for any $\lambda\in\Lambda$, 
\[
\begin{split}
    \mathbb{E}[t_i.O] - \mathbb{E}[\lambda(t_i).O] &= \mathbb{E}[t_i.O - \lambda(t_i).O]\\
    &\le \mathbb{E}[|t_i.O - \lambda(t_i).O|] \le d_o.
\end{split}
\]
Hence $\mathbb{E}[\lambda(t_i).O] \ge \mathbb{E}[t_i.O] - d_o$. By assumption, \sskr returns a non-NULL package $x^*$, and the corresponding objective value for $y=\lambda(x^*)$ in $Q(S_0)$ is
\[
\begin{split}
    \omega^*_s &= \sum_{j=1}^d \mathbb{E}[r_j.O]y_j \\
    &= \sum_{i=1}^{n} \mathbb{E}[\lambda(t_i).O]x^*_i \ge \sum_{i=1}^{n} (\mathbb{E}[t_i.O]-d_o)x^*_i\\
    &= \sum_{i=1}^{n} \mathbb{E}[t_i.O]x^*_i - d_o \sum_{i=1}^{n} x^*_i \ge \omega^* - d_o P_{max}.
\end{split}
\]
Since $x^*$ is assumed to be sketch feasible, we can choose $\lambda$ so that $y=\lambda(x^*)$ is a feasible solution of $Q(S_0)$. Then the objective value $\omega^*_{sk}$ for the optimal solution $y^*$ to $Q(S_0)$ satisfies
$\omega^*_{sk} \ge \omega^*_s\ge\omega^* - d_o P_{max}$. Note that $\omega^*_{sk}  \ge \omega^*_{s}$ may hold even if $x^*$ is not sketch feasible, in which case the approximation guarantee will still hold. As discussed before the statement of the theorem, the objective value $\Bar{\omega}$ for the solution to $Q(S_0)$ returned by \textsc{SolveSketch} satisfies $\Bar{\omega} \ge (1-\epsilon)\omega^*_{sk}\ge (1-\epsilon)(\omega^* - d_o P_{max})$. Moreover, the objective value $\omega$ corresponding to the solution $x^*$ satisfies $\omega \ge (1-\epsilon) \Bar{\omega}$, so that
%
\begin{align*}
\omega \ge (1-\epsilon) \Bar{\omega} \ge (1-\epsilon)^2(\omega^* - d_o P_{max})    
\end{align*}
as desired.
\end{proof}

The probability that the optimal package is sketch feasible increases as the diameter thresholds for the attributes in the query constraints become
small. Small diameter thresholds ensure that a package's sum for any attribute does not change significantly when switching to the sketch problem; see 
\ref{appendix:sketch_feasibility}
for details. We can obtain similar guarantees on minimization problems.

\section{Probability of the optimal package being sketch feasible}
\label{appendix:sketch_feasibility}

Theorem~\ref{th:opt_sskr} assumes that the optimal package solution to the SPQ of interest is sketch-feasible, as defined prior the theorem. To provide insight into sketch feasibility, we present Theorem~\ref{theorem:sketch_feasibility}, which analyzes the probability of the optimal package being sketch feasible with respect to a given \LCVaR constraint. 

\begin{theorem}
\label{theorem:sketch_feasibility}
Denote by $x^*=\{x^*_1, \dots, x^*_n\}$ a non-NULL optimal feasible package computed by \lcsolve or \sskr for a query $Q(S)$ where $S=\{t_1, \dots, t_n\}$, and denote by $s$ the package cardinality. Consider an \LCVaR constraint $\sum_{i=1}^{n} CVaR_{\alpha}(t_i.C)x_i \ge V$ and suppose that
\begin{itemize}
\item[(i)] $d_p\ge\sum_{i\in I_p}I(x^*_i>0)$ for each partition $p$ in the sketch $S_0$ of $S$, where $d_p$ is the number of duplicate sketch representatives in partition $p$, $I$ is an indicator function, and $I_p$ is the set of indices for all tuples $t_i\in S$ that lie in partition~$p$; and
\item[(ii)] the CVaR values (with respect to attribute $C$) for the tuples in each partition are independently and uniformly distributed around the CVaR value of their representative.
\end{itemize}
Then the probability of $x^*$ being sketch feasible is
\[
1 - \frac{1}{s!}\sum_{k=0}^{\lfloor (\frac{s}{2}-\frac{\delta\alpha}{d_C}) \rfloor} (-1)^k{s \choose k} \left(\frac{s}{2}-\frac{\delta\alpha}{d_C} - k\right)^s,
\]
where $\delta=\sum_{i=1}^{n} \text{CVaR}_{\alpha}(t_i.C)x^*_i - V$ and $d_C$ is the diameter threshold for $C$.
\end{theorem}
\begin{proof}
Suppose a tuple $t_i$ is replaced by $\lambda(t_i)$, a duplicate of their representative in the sketch problem. The diameter constraint dictates that the mean absolute distance between $t_i$ and $\lambda(t_i)$ w.r.t. $C$ is bounded by a threshold $d_C$, i.e., $\mathbb{E}[|\lambda(t_i.C)-t_i.C|] \le d_C$. From Theorem~\ref{th:mad_cvar}, we have 
\[
|\text{CVaR}_{\alpha}(\lambda(t_i.C)) - \text{CVaR}_{\alpha}(t_i.C)| \le \frac{d_C}{\alpha}.
\]
Assumption~(ii) implies that 
\[
\text{CVaR}_{\alpha}(t_i.C) = \text{CVaR}_{\alpha}(\lambda(t_i.C)) + \mathcal{U}_i\left[-\frac{d_C}{\alpha}, \frac{d_C}{\alpha}\right],
\]
where $\mathcal{U}_i[l, r]$ is an independent sample from the uniform distribution over the interval $[l,r]$. Since, by Assumption~(i), every tuple is mapped to a distinct representative (to minimize risk) and since $\delta\ge 0$ by the assumed feasibility of $x^*$ for $Q(S)$,
the probability of $x^*$ being feasible with respect to the constraint as defined over $S_0$ is 
\begin{align*}
& \mathbb{P}\left(\sum_{i=1}^{n} \text{CVaR}_{\alpha}(\lambda(t_i.C))x^*_i \ge V\right)\\
&\quad = \mathbb{P}\left(\sum_{i=1}^{n} \left(\text{CVaR}_{\alpha}(t_i.C) + \mathcal{U}_i\left[-\frac{d_C}{\alpha}, \frac{d_C}{\alpha}\right]\right)x^*_i \ge V\right)\\
&\quad = \mathbb{P}\left(\sum_{i=1}^{n} \text{CVaR}_{\alpha}(t_i.C)x^*_i + \sum_{i=1}^{n}\left(2\frac{d_C}{\alpha}\mathcal{U}_i[0, 1] - \frac{d_C}{\alpha}\right)x^*_i \ge V\right)\\
&\quad = \mathbb{P}\left(\sum_{i=1}^{n} \text{CVaR}_{\alpha}(t_i.C)x^*_i + 2\frac{d_C}{\alpha}\sum_{i=1}^{n}\mathcal{U}_i[0, 1]x^*_i - \frac{d_C}{\alpha}\sum_{i=1}^{n}x^*_i \ge V\right)\\
&\quad = \mathbb{P}\left(V + \delta + 2\frac{d_C}{\alpha} \sum_{i=1}^{s}\mathcal{U}_i[0, 1] - \frac{d_C s}{\alpha} \ge V\right)\\
&\quad = \mathbb{P}\left(2\frac{d_C}{\alpha} \sum_{i=1}^{s}\mathcal{U}_i[0, 1] \ge \frac{d_C s}{\alpha}-\delta\right)\\
&\quad = \mathbb{P}\left(\sum_{i=1}^{s}\mathcal{U}_i[0, 1] \ge \frac{s}{2}-\frac{\delta\alpha}{2d_C}\right)\\
&\quad = 1 - \mathbb{P}\left(\sum_{i=1}^{s}\mathcal{U}_i[0, 1] \le \frac{s}{2}-\frac{\delta\alpha}{2d_C}\right)\\
&\quad = 1 - \frac{1}{s!}\sum_{k=0}^{\lfloor (\frac{s}{2}-\frac{\delta\alpha}{2d_C}) \rfloor} (-1)^k{s \choose k} \left(\frac{s}{2}-\frac{\delta\alpha}{2d_C} - k\right)^s
\end{align*}
The last line follows from the the Irwin-Hall formula for the CDF of a sum of independent uniform random variables.
\end{proof}

The result can be trivially modified to find the probability of $x^*$ being sketch-feasible w.r.t. \LCVaR constraints of different formats, and can be used for finding the probability of sketched version of $x^*$ satisfying a deterministic constraint by setting $\alpha = 1$. To find the overall probability of $x^*$ being sketch feasible, we can make a simplifying assumption that the probability of satisfying each constraint is independent of each other, and multiply the probabilities together or, alternatively, use the Bonferroni inequality to derive a conservative probability estimate.

For any feasible package, $\delta \ge 0$ is guaranteed for every constraint. The probability of any constraint being satisfied is the least when $\delta = 0$. In that extremity, the probability of the optimal package being sketch-feasible w.r.t. that constraint is $1/2$, regardless of the value of $s$. For a query containing $C$ such mutually independent constraints, the probability of the optimal package being sketch feasible is thus $1/2^C$.

Theorem~\ref{theorem:sketch_feasibility} assumes that \lcsolve is able to find the optimal package by replacing each risk constraint with an appropriately parameterized \LCVaR constraint. While this is a strong assumption, \lcsolve usually terminates with a package that has an objective value of at least $(1-\epsilon)\omega^{*}$. The probability of this package being sketch-feasible can then be derived from Theorem~\ref{theorem:sketch_feasibility}. In the event that this package is indeed sketch-feasible, the objective value of the package produced consequently by \sskr is lower bounded by the result of Theorem~\ref{th:opt_sskr}.

\section{Notes on Setting Hyperparameters}
\label{sec:hyper}
\revc{We provide notes on how a user can set the value of each hyperparameter used in \lcsolve and \sskr.

Like their predecessor \ssearch, they require the specification of the initial number of optimization scenarios ($m$). A key consideration while setting $m$ is the time needed to generate each scenario. This depends on the computational efficiency of the Variable Generator (VG) function, since optimization scenarios require that values be generated for every tuple. For more expensive VG functions, scenario generation times can be high, and users may wish to start with a smaller number of optimization scenarios in hopes that feasible and near-optimal packages can be generated from them. If not, \lcsolve and \sskr will automatically increase the number of optimization scenarios and retry. In our experiments, we used a relatively simple set of VG functions, i.e., Geometric Brownian Motion for the Portfolio dataset and Gaussian noises for the TPC-H dataset. We used $100$ as our default value for $m$, but also found that setting $m$ to $50$ or $200$ did not largely effect the overall runtimes or relative integrality gap of packages.

We further require the specification of the number of validation scenarios $\hat{m}$. We recommend $\hat{m}$ to be set as large as possible, unless validation scenario generation runtimes become prohibitively expensive. The need for having more validation scenarios increases with the variance of the stochastic attributes. In our experiments, we set $\hat{m}=10^6$. To validate if this is a judicious choice for our datasets, we ran an experiment where we ran $30$ trials for finding the objective value of $30$ randomly generated packages containing $5$ distinct tuples each over $10^6$ i.i.d. scenarios. The objective values of each package were very uniform across the trials, with average variances of objective value measurements being in the orders of $10^{-35}$ and $10^{-29}$ for the Portfolio and TPC-H datasets respectively. This indicates that using $10^6$ validation scenarios can give us a stable measurement of the objective value of a package in these datasets. Readers can repeat such experiments in their own setups to determine appropriate values of $\hat{m}$.

\lcsolve and \sskr both require the specification of the approximation error bound, $\epsilon$. In response, except in very rare cases which never occurred during our experiments, \lcsolve has a $(1\pm \epsilon)$-guarantee on the relative integrality gap of its generated packages, while \sskr usually provides a similar guarantee in the order of $(1\pm\epsilon)^2$. In our experiments, we set $\epsilon$ to $0.05$. We noted that many of the returned packages had better relative integrality gaps than what would be expected from this bound, since in practice, \lcsolve tends to terminate its Alternating Parameter Search with near-optimal parameterizations that result in better packages than indicated by $\epsilon$. For all queries in our workload, relaxing $\epsilon$ from $0.05$ to $0.1$, does not result in the relative optimality gap increasing by more than $0.01$, and does not cause a significant reduction in runtime. We note, however, that setting $\epsilon$ to $0.01$ can make the optimality requirement too strict, with no package result being found within an hour of execution by \sskr for Portfolio's queries Q3-Q5. In new environmental setups, users can try a similar search for an appropriate value of $\epsilon$, relaxing it for queries where a package solution cannot be found within a reasonable amount of time. 

For all the queries in our workloads, the initial risk tolerance $\Gamma$ was set to over $0.99$, so our choice of setting $\Delta \Gamma = 0.03$ had no effect on the results. Similarly, changing the bisection termination threshold $\delta$ from $1e-2$ to $5e-3$ and $5e-2$ did not have a noticeable impact on the runtime or quality of the generated packages for any query. For the \dpart-related parameters, as mentioned in Section~\ref{sec:experimental_evaluation}, the size threshold $\tau$ was set to $100,000$ based on experiments where we saw Gurobi solving randomly generated ILPs with $10^5$ variables and $3$ linear constraints in about a minute. To set the diameter constraints, we tried different combinations of values, and chose the one where the number of resulting partitions was within $[\frac{\tau}{10}, \frac{\tau}{2}]$, and for which the diameter constraints for every attribute were reasonably tight.
}

\section{Details of the Experiments}
\subsection{Datasets}
\label{appendix:datasets}

We varied the TPC-H dataset for different comparisons as follows:

\begin{itemize}
    \item To demonstrate scalability over increasing relation sizes, we created different versions of the relation having $20{,}000$, $60{,}000$, $120{,}000$,  $300{,}000$, $450{,}000$, $600{,}000$, $1{,}200{,}000$, $3{,}000{,}000$, $4{,}500{,}000$ and $6{,}000{,}000$ tuples.
    \item To compare between \ssearch and \lcsolve, we created a dataset consisting on $20,000$ randomly selected tuples. To test the effect of increased variances on the approaches, we created alternate versions of the relation where the variances of each tuple where multiplied by factors of $1, 2, 5, 8, 10, 13, 17$ and $20$.
\end{itemize}

\noindent We varied the portfolio dataset as follows:

\begin{itemize}
    \item To demonstrate scalability over an increasing number of tuples, we varied the intervals after which stocks could be sold from $3$ months to $45$ days, $1$ month, $15$ days, $9$ days, $3$ days, $1$ day and $0.5$ days respectively.
    \item For \ssearch vs \lcsolve, we used the smallest dataset created by restricting the interval after which stocks could be sold to $3$ months. To test differing variances, we boosted up the volatility of each stock by factors of $1, 2, 5, 8, 10, 13, 17$ and $20$.
\end{itemize}

\subsection{Queries}
\label{appendix:hardness}
\paragraph{Estimating query hardness.} A query is `hard' if the probability of a random package being feasible w.r.t. it is small. We estimate this probability similar to the method introduced in~\cite{mai2023scaling}. The probability that a package with $\mathbf{s}$ tuples satisfies a constraint of the form $\sum_{i=1}^{n} A_{i}x_i \le V$, where the values of $A$ are normally distributed with mean $\mu$ and variance $\sigma^2$, is equivalent to the probability of $\sum_{i=1}^{s} N(\mu, \sigma^2) \le V$, or $N(\mu s, \sigma^2 s) \le V$ according to the central limit theorem. This probability is given by the CDF of the normal distribution $N(\mu s, \sigma^2 s)$ at $V$. For any given query, we make the simplifying assumption that the probability of each constraint being satisfied is independent of any other constraint being satisfied. Hence, if a query has $C$ constraints $C_1$, $C_2$, $\dots$, $C_C$ with satisfaction probabilities $P(C_1), \dots, P(C_C)$, the probability of the query being satisfied is $\prod_{i=1}^{C}P(C_i)$. We define hardness as the negative log likelihood of this probability, i.e., $H(Q)=-\log{\prod_{i=1}^{C}P(C_i)}$. Given an SPQ, we first relax it by removing all integrality constraints from it, solve it using \lcsolve and measure the hardness of the ILP that finds a feasible and near-optimal package.

\vskip0.8\baselineskip
\label{appendix:workload}
\noindent \textbf{Portfolio Queries:}
\vskip0.2\baselineskip
\noindent Q1:

\begin{lstlisting}[basicstyle=\footnotesize\ttfamily\bfseries]
SELECT PACKAGE(*) AS P
FROM Stock_Investments_Half
SUCH THAT
COUNT(*) <= 30 AND 
SUM(Price) <= 500 AND
SUM(Gain) >= 350 WITH PROBABILITY >= 0.95
MAXIMIZE EXPECTED SUM(Gain)
\end{lstlisting}

\noindent Q2:

\begin{lstlisting}[basicstyle=\footnotesize\ttfamily\bfseries]
SELECT PACKAGE(*) AS P
FROM Stock_Investments_Half
SUCH THAT
COUNT(*) <= 30 AND 
SUM(Price) <= 1000 AND
SUM(Gain) >= 600 WITH PROBABILITY >= 0.97
MAXIMIZE EXPECTED SUM(Gain)
\end{lstlisting}

\noindent Q3:

\begin{lstlisting}[basicstyle=\footnotesize\ttfamily\bfseries]
SELECT PACKAGE(*) AS P
FROM Stock_Investments_Half
SUCH THAT
COUNT(*) <= 30 AND 
SUM(Price) <= 1000 AND
SUM(Gain) >= 900 WITH PROBABILITY >= 0.97
MAXIMIZE EXPECTED SUM(Gain)
\end{lstlisting}

\noindent Q4:

\begin{lstlisting}[basicstyle=\footnotesize\ttfamily\bfseries]
SELECT PACKAGE(*) AS P
FROM Stock_Investments_Half
SUCH THAT
COUNT(*) <= 30 AND 
SUM(Price) <= 1000 AND
SUM(Gain) >= 900 WITH PROBABILITY >= 0.97 AND
SUM(Gain) >= 1000 WITH PROBABILITY >= 0.90
MAXIMIZE EXPECTED SUM(Gain)
\end{lstlisting}

\noindent Q5:

\begin{lstlisting}[basicstyle=\footnotesize\ttfamily\bfseries]
SELECT PACKAGE(*) AS P
FROM Stock_Investments_Half
SUCH THAT
COUNT(*) <= 30 AND 
SUM(Price) <= 1000 AND
SUM(Gain) >= 900 WITH PROBABILITY >= 0.97 AND
SUM(Gain) >= 1500 WITH PROBABILITY >= 0.90
MAXIMIZE EXPECTED SUM(Gain)
\end{lstlisting}

\noindent\textbf{TPC-H Queries:}

\noindent Q1:

\begin{lstlisting}[basicstyle=\footnotesize\ttfamily\bfseries]
SELECT PACKAGE(*) AS P
FROM Lineitem_6000000
SUCH THAT 
COUNT(*) <= 30 AND 
SUM(Tax) <= 0.05 AND 
SUM(QUANTITY) <= 20 WITH PROBABILITY >= 0.95 AND 
SUM(PRICE) >= 750 WITH PROBABILITY >=0.90
MAXIMIZE EXPECTED SUM(PRICE)
\end{lstlisting}
\clearpage
\noindent Q2:

\begin{lstlisting}[basicstyle=\footnotesize\ttfamily\bfseries]
SELECT PACKAGE(*) AS P
FROM Lineitem_6000000
SUCH THAT 
COUNT(*) <= 30 AND 
SUM(Tax) <= 0.03 AND 
SUM(QUANTITY) <= 10 WITH PROBABILITY >= 0.95 AND 
SUM(PRICE) >= 750000 WITH PROBABILITY >=0.95
MAXIMIZE EXPECTED SUM(PRICE)
\end{lstlisting}

\noindent Q3:

\begin{lstlisting}[basicstyle=\footnotesize\ttfamily\bfseries]
SELECT PACKAGE(*) AS P
FROM Lineitem_6000000
SUCH THAT 
COUNT(*) <= 30 AND 
SUM(Tax) <= 0.02 AND 
SUM(QUANTITY) <= 8 WITH PROBABILITY >= 0.95 AND 
SUM(PRICE) >= 7500000 WITH PROBABILITY >=0.97
MAXIMIZE EXPECTED SUM(PRICE)
\end{lstlisting}

\noindent Q4:

\begin{lstlisting}[basicstyle=\footnotesize\ttfamily\bfseries]
SELECT PACKAGE(*) AS P
FROM Lineitem_6000000
SUCH THAT 
COUNT(*) <= 30 AND 
SUM(Tax) <= 0.02 AND 
SUM(QUANTITY) <= 8 WITH PROBABILITY >= 0.98 AND 
SUM(PRICE) >= 750000000 WITH PROBABILITY >=0.98
MAXIMIZE EXPECTED SUM(PRICE)
\end{lstlisting}

\noindent Q5:

\begin{lstlisting}[basicstyle=\footnotesize\ttfamily\bfseries]
SELECT PACKAGE(*) AS P
FROM Lineitem_6000000
SUCH THAT 
COUNT(*) <= 30 AND 
SUM(Tax) <= 0.01 AND 
SUM(QUANTITY) <= 5 WITH PROBABILITY >= 0.98 AND 
SUM(PRICE) >= 1000000000 WITH PROBABILITY >=0.98
MAXIMIZE EXPECTED SUM(PRICE)
\end{lstlisting}
\subsection{\revbc{Example Package}}
\label{sec:ex_package}
\revbc{As a case study, we show the package obtained by employing \sskr to solve the query Q1 for the full-sized Portfolio Dataset which contains one tuple for each stock investment option where company shares can be sold at intervals of half a day (i.e. the relation in Figure~\ref{fig:portfolio-builder}). Table~\ref{tab:pf-q1} shows that the resulting package invests in a multiple stock options, which helps to satisfy the risk constraint in the query. It is easy to verify that the package satisfies the budgetary constraint of being within $\$500$, requiring a total price of $499.19\$$. The two major investments recommended by the package, MCO and XCLS showed strong growth over 2023-2025, indicating that portfolios recommended by \sskr may be useful in real-world stock markets.

We, note, however, that these packages were created by generating gains using Geometric Brownian Motion, and the usage of more sophisticated scenario generation models could lead to better portfolios generated by \sskr. The example portfolio should not be considered as financial advice.}

\begin{table}[h]
\caption{\revbc{Portfolio Created by \sskr in response to Q1 over 4.8 million tuples}}
\label{tab:pf-q1}
\begin{tabular}{@{}lclc@{}}
\toprule
Ticker & Sell After (Days) & Price  & Multiplicity \\ \midrule
MCO    & 278.5             & 289.50 & 1            \\
EXLS   & 648.5             & 180.24 & 1            \\
STBZ   & 341               & 21.59  & 1            \\
HNTIF  & 587               & 3.33   & 2            \\
EQS    & 614.5             & 1.65   & 1            \\ \bottomrule
\end{tabular}
\end{table}

\begin{table*}[]
\caption{\revbc{We compared the relative integrality gap (a lower bound of how close the objective value of a package is with the optimal package) of the generated packages from Q1 of the Portfolio dataset when representatives are (a) duplicated, (b) not duplicated and (c) not duplicated, but multiple tuples are chosen as representatives from the same partition. The objective values of packages unsurprisingly fall sharply when a representative is not duplicated at all. Taking multiple randomly chosen representatives from the same partition instead of duplicating the same single representative yields slightly more suboptimal packages than forming correlated duplicates of the same representative.}}
\label{tab:ablation_dups}
\begin{tabular}{llll}
              & \multicolumn{3}{c}{Relative Integrality Gap}                         \\ \hline
Relation Size & With Duplicates & Without Duplicates & With Multiple Representatives \\ \hline
161,161       & \textbf{0.05}            & 0.82               & 0.06                          \\
269,698       & \textbf{0.06}            & 0.94               & 0.06                          \\
802,516       & \textbf{0.06}            & 0.97               & 0.07                          \\
2,400,970     & \textbf{0.04}            & 0.93               & 0.09                          \\
4,801,940     & \textbf{0.07}            & 0.95               & 0.09                         
\end{tabular}
\end{table*}

\section{Additional Experiments}\label{appendix:addlExpts}

\revbc{\subsection{Ablation Experiment: Use of Duplicates}\label{appendix:duplicatesAblation}
\label{sec:ablation_dup}
We compared the relative integrality gap, as defined in Section~\ref{sec:experimental_evaluation}, of the objective values of the generated packages when representatives are duplicated, to those produced when they are not duplicated at all, and when multiple unduplicated representatives are chosen from the same partition. The results are compiled in Table~\ref{tab:ablation_dups}. Lower relative optimality gaps imply better packages.

Without duplication, as mentioned in Section~\ref{section:sskr}, packages with multiple tuples from the same partition are often rendered validation-infeasible, and \sskr is thus limited to choosing the best out of the more suboptimal packages that remain feasible, causing a large relative optimality gap.

A second alternative we tested was to choose a random subset of $\min(P_{max}, m)$ tuples from each partition, where $m$ represents the maximum number of tuples in a partition, and $P_{max}$ represents an upper bound on the number of tuples in a package. In this setting, the tuples in the partition that are not chosen as representatives have a lower representation during sketch than our original proposal of using correlated duplicates.

For example, suppose that our random sample contains three representatives from a partition, all of which are highly correlated with each other. Among the tuples in that partition that are not chosen in the sample, there may be some that are not positively correlated with any of the representatives. However, the sketch solver only has access to the three highly correlated representatives, and thus may regard the partition as `too risky' to choose multiple tuples from, and opt to create suboptimal package with tuples from other partitions instead. This is similar to the situation with only one representative from the partition with multiplicity $>1$. 

In our proposal, a single `median-like' representative is duplicated, with each duplicate's Pearson's correlation coefficient with the others being equal to the median of the representative's correlation with all other tuples in the partition. As such, if there exist tuples that are not highly correlated with the representative, the duplicates will also exhibit lower inter-tuple correlation between themselves in the sketch problem; this avoids the potential issue of the sketch solver discarding the partition due to risk constraints being violated by choosing several highly correlated tuples. This phenomenon is seen in Table~\ref{tab:ablation_dups}, where packages created from multiple correlated duplicates have a lower optimality gap than those created from multiple randomly chosen un-duplicated representatives.\\}

\subsection{\revb{Ablation Experiment: Partitioning Schemes}}
\label{sec:ablation_dpart}
%
\revb{In this experiment, we compare \dpart to three alternative schemes, based on off-the-shelf methods:
\begin{itemize}
    \item \emph{\kd}: It is used by \skr in deterministic settings~\cite{brucato2015scalable}, and modified to use MAD as the distance measure.
    \item \emph{PCA+AC}: We represent each tuple as a high-dimensional vector constructed by concatenating the scenario values for all the stochastic attributes with the values of every deterministic attribute; we then use PCA~\cite{abdi2010principal} to reduce the number of dimensions to the number of attributes, and finally use single linkage agglomerative clustering~\cite{tokuda2022revisiting} to partition the tuples.
    \item \emph{FINDIT}: We represent tuples as in PCA+AC, but now partition the vectors without reducing their dimensionality using FINDIT~\cite{woo2004findit},
    a fast, high-dimensional clustering algorithm.\footnote{For both PCA+AC and FINDIT, we cluster vectors based on Euclidean distance. We set their hyperparameters to ensure the number of resulting partitions remained in the range $\{\frac{\tau}{10}, \frac{\tau}{2}\}$, similar to determining the diameter thresholds for \dpart.}
\end{itemize}

Figure~\ref{fig:dpart} reports the offline pre-processing time and solution quality---measured as the integrality gap---for each query executed over the resulting partitioning. \kd is slightly faster than \dpart (at most a 5 minute decrease in runtime), but the result quality is inferior. 
Since \kd does not account for shapes of probability distributions as \dpart does, it cannot provide theoretical guarantees of intra-partition tuple similarity, which translates to less optimal packages. PCA+AC and FINDIT 
require much larger runtimes. Their integrality gaps are also higher: the 
aggressive dimensionality reduction in PCA+AC tends to obfuscate inter-tuple differences in tail behaviors for a given stochastic attribute,
and the dimension voting in FINDIT (which clusters points with low differences across a large number of dimensions) ignores large differences that may occur in a small number of tail scenarios.
}

\begin{figure*}[!h]
\renewcommand{\arraystretch}{0.8}
\begin{tabular}{lccccccccccccccccc}
\toprule
                 & \multicolumn{5}{c}{{\footnotesize Offline pre-computation (mins)}}                     &                           & \multicolumn{11}{c}{{\footnotesize Relative Integrality Gap}}                            \\
                 \cmidrule{2-6}\cmidrule{8-18}
                 & \multicolumn{2}{c}{{\footnotesize partitioning time}} && \multicolumn{2}{c}{{\footnotesize preprocessing time}} & & \multicolumn{5}{c}{{\footnotesize Portfolio}}   & & \multicolumn{5}{c}{{\footnotesize TPC-H}}        \\
                 \cmidrule{2-3}\cmidrule{5-6}\cmidrule{8-12}\cmidrule{14-18}
                 & {\footnotesize Portfolio}              & {\footnotesize TPC-H}          &    & {\footnotesize Portfolio}               & {\footnotesize TPC-H}            &  
                 & {\footnotesize Q1}   & {\footnotesize Q2}   & {\footnotesize Q3}   & {\footnotesize Q4}   & {\footnotesize Q5} &  & {\footnotesize Q1}   & {\footnotesize Q2}   & {\footnotesize Q3}   & {\footnotesize Q4}   & {\footnotesize Q5}   \\
                 \midrule
\dpart & 4.60                   & 18.70           &    & 44.42                   & 27.50           &   & \textbf{0.03} & \textbf{0.02} & \textbf{0.04} & \textbf{0.05} & \textbf{0.04} & & \textbf{0.03} & \textbf{0.05} & \textbf{0.04} & \textbf{0.05} & \textbf{0.04} \\
KD Trees         & 3.85                   & 14.12           &    & 43.72                   & 23.42           &   & 0.07 & 0.08 & 0.11    & 0.13 & 0.12   & & 0.07 & 0.07 & 0.06 & 0.08 & 0.09   \\
\revb{PCA+AC} & \revb{52.32}        &  \revb{128.49}        &     &  \revb{93.56}                 & \revb{136.93}     &     &  \revb{0.05} & \revb{0.06} & \revb{0.09} & \revb{0.12} & \revb{0.11} & & \revb{0.08} & \revb{0.07} & \revb{0.06} & \revb{0.06} & \revb{0.08} \\ 
\revb{FINDIT} & \revb{36.62} &    \revb{28.31}  & & \revb{77.23}  & \revb{36.75} &  & \revb{0.04} & \revb{0.03} & \revb{0.04} & \revb{0.06} & \revb{0.05} & & \revb{0.07} & \revb{0.07} & \revb{0.06} & \revb{0.08} & \revb{0.07}\\
\bottomrule
\end{tabular}
\vspace{-3mm}
\caption{Offline pre-computation runtime, which includes representative selection, is minimally higher with \dpart compared to KD-trees, but the resulting packages achieve lower relative integrality gaps when \dpart is used. \revb{Off-the-shelf techniques such as dimensionality reduction with PCA and Agglomerative Clustering, and high-dimensional clustering approaches like FINDIT have larger offline pre-computation runtimes, and generate packages with larger integrality gaps}
}
\label{fig:dpart}
\end{figure*}

\smallskip
\noindent
\fbox{
\parbox{0.96\columnwidth}{
\emph{Key takeaway:} 
   \dpart results in better intra-partition tuple similarity, leading to more optimal packages.
}}

\subsection{\revc{Robustness of Generated packages}}\label{appendix:robustness}

\revc{For every package generated during our experiments by \lcsolve, \ssearch and \sskr, we tested their robustness by checking if they satisfy the risk constraints over a set of $10^6$ test scenarios. The test scenarios were never used for optimization/validation during the intermediate stages of the package generation process by any of the algorithms. All the validation-feasible packages produced during all our experiments satisfied every risk constraint among these test scenarios.}





\end{document}